\documentclass[12pt]{article}

\usepackage[T1]{fontenc}
\usepackage[english]{babel}
\usepackage{amsmath,amsfonts,amssymb,amsthm}
\usepackage{mathabx}
\usepackage{graphicx,xcolor}
\usepackage[round,authoryear]{natbib}
\usepackage[colorlinks=true, allcolors=blue]{hyperref}
\usepackage{bm}
\usepackage{booktabs}
\usepackage{float}

\newcommand{\E}{\mathbb{E}}
\newcommand{\R}{\mathbb{R}}
\newcommand{\N}{\mathbb{N}}
\newcommand{\cov}{\mathrm{Cov}}
\renewcommand{\vec}{\mathrm{vec}}
\newcommand{\mat}{\mathrm{mat}}
\newcommand{\tr}{\operatorname{tr}}
\newcommand{\ot}{\,\tilde\otimes\,}
\newcommand{\mmle}{\widehat \Sigma_1 \ot \widehat \Sigma_2}
\newcommand{\pto}{\stackrel{p}{\to}}
\newcommand{\op}{\mathrm{op}}
\DeclareMathOperator*{\argmin}{arg\,min}

\newtheorem{theorem}{Theorem}
\newtheorem{definition}{Definition}
\newtheorem{proposition}{Proposition}
\newtheorem{lemma}{Lemma}
\newtheorem{corollary}{Corollary}
\newtheorem{remark}{Remark}

\title{Testing Covariance Separability in High Dimensions}
\author{
Tomas Masak $^{\dagger}$ \and
Marcus Mayrhofer $^{\ddagger}$ \and
Una Radojičić $^{\ddagger}$ \\[0.6em]
\small $^{\dagger}$Institute for Statistics and Mathematics, WU Wien, Vienna, Austria \\
\small $^{\ddagger}$ Institute of Statistics and Mathematical Methods in Economics, TU Wien, Vienna, Austria
}

\hypersetup{
	pdftitle={Testing Covariance Separability in High Dimensions},
	pdfsubject={stat.ME},
	pdfauthor={Tomas~Masak, Marcus~Mayrhofer, Una~Radojičić},
	pdfkeywords={Matrix-variate data, matrix normal MLE, elliptical distribution, Kronecker covariance},
}

\begin{document}
\maketitle

\begin{abstract}
Separability is an important structural assumption often placed on the covariance when working with matrix-variate data, because it greatly simplifies both interpretation and computation of subsequent covariance-based statistical tasks. Yet testing the separability assumption is difficult in the high-dimensional regime. We propose to test separability by recasting the problem as a sphericity test after whitening the data using the separable maximum likelihood estimate of the covariance. The test is calibrated by Monte Carlo simulation, yielding finite-sample level control. Furthermore, we prove the test's high-dimensional consistency under dense alternatives. To reduce its reliance on distributional assumptions, we introduce an angular version of the test based on radial normalization after whitening. We demonstrate the practical utility, empirical power, and computational efficiency of the proposed tests in a large simulation study and on a real-world acoustic data set.
\end{abstract}


\section{Introduction}

Contemporary statistical data sets increasingly often arrive in forms that are naturally indexed by more than a single coordinate. Images \citep{wu2023}, spatio-temporal fields \citep{gneiting2002}, longitudinal panels \citep{ding2018matrix}, neuroscience signals~\citep{masak2023separable}, graph or network observations \citep{zhouGemini}, electricity load or traffic volume profiles \citep{bouveyron2018,chiou2014}, spectrograms \citep{pigoli2018statistical}, and other time-frequency representations or relational data \citep{volfovsky2015} are all examples in which a single observation is more faithfully represented as a matrix than as a vector. Even data that are initially recorded as one-dimensional signals may acquire a matrix structure during standard pre-processing. A speech recording, for instance, is a waveform at the acquisition stage, but it is commonly transformed into a log-spectrogram or into a matrix of mel-frequency cepstral coefficients, where one coordinate indexes frequency-type features, and the other indexes localized windows of the recording.
Flattening such data into long vectors is always possible, but it discards the natural structure that is potentially both scientifically meaningful and computationally beneficial.

Regardless of the specific form of the samples, covariance remains the central object of statistical analysis. Since it encodes second-order dependence, it underlies dimension reduction, regression, or classification tasks, and is often the central object through which regularization is imposed, be it by assuming a specific structure of the covariance, or utilizing shrinkage. In the matrix-variate setting, however, the covariance is no longer an ordinary matrix of a moderate size. If $X \in \R^{p \times q}$, the covariance is a fourth-order tensor $\Sigma \in \R^{p \times q \times p \times q}$ indexed by pairs of row and column coordinates. Thus, an unrestricted covariance contains $\mathcal{O}(p^2q^2)$ entries. Already for moderately sized $p$ and $q$, this creates two interlinked difficulties. Statistically, the number of covariance parameters is typically much larger than the number of available observations $N$, making unstructured estimation unstable, if not entirely impossible. Computationally, explicitly forming, storing, or subsequently manipulating (e.g.~inverting, for the purposes of prediction) with the empirical covariance is prohibitive.

This tension necessitates structural assumptions that induce parsimony while avoiding a fixed, finite-dimensional parametric specification of the covariance model. Examples include (inverse) sparsity, stationarity, or low-rank/factor structures. Another example specific to matrix-valued data is the separability of the covariance. A covariance operator $\Sigma \in \R^{p \times q \times p \times q}$ of a random matrix $X \in \R^{p \times q}$ is separable if it factorizes as
\[
\Sigma = \Sigma_1 \ot \Sigma_2,
\]
for two positive definite matrices $\Sigma_1 \in \R^{p \times p}$ and $\Sigma_2 \in \R^{q \times q}$. The factorization is defined on the level of entries as
\[
\Sigma_{i,j,k,l} = (\Sigma_1)_{i,k}(\Sigma_2)_{j,l}, \qquad i,k=1,\ldots,p,\; j,l=1,\ldots,q.
\]
The appeal of separability is immediate: the unrestricted covariance is replaced by two much smaller covariance factors, one capturing the covariance structure between the rows and the other between the columns of the matrix. This reduces the number of parameters from $\mathcal{O}(p^2q^2)$ to $\mathcal{O}(p^2+q^2)$. On a standard laptop with 64 GB of memory, an unstructured covariance with $p=q\approx300$ consumes almost all the memory, while a separable covariance with $p=q=10\,000$ is still manageable.

For these reasons, separability has a long history in data analysis; see for instance the early work of \citet{dawid1981matrix} and \citet{mardia1993spatiotemporal}. More recently, separable covariance models have attracted renewed attention in high-dimensional statistics \citep{greenewald2019tensor,hoff2023core, zhang2023matrix,yu2023testing,mccormack2025information}. In particular, the separable maximum likelihood estimator was shown to exist and be well behaved in regimes where the ordinary empirical covariance is singular or severely ill-conditioned \citep{dutilleul1999mle,soloveychik2016gaussian,drton2021,franks2026near}.

However, the advantages of separability come at a cost. Separability rules out interactions between the two indexing directions and forces the dependence across one direction to be modulated only by a scalar factor depending on the other \citep{gneiting2002,gneiting2006,rougier2017}. In applications, this can lead to oversimplification. If separability is imposed incorrectly, the resulting covariance estimate may lose relevant dependence information as well as distort and bias downstream statistical tasks such as prediction, classification, or uncertainty quantification \citep{aston2017tests,masak2023random}. Therefore, the assumption should not be imposed blindly.

This makes testing separability an essential task, one that should precede many matrix-variate data analyses. A test of
\[
H_0: \Sigma = \Sigma_1 \ot \Sigma_2 \qquad \text{vs.} \qquad H_1: \Sigma \neq \Sigma_1 \ot \Sigma_2 \text{ for any } \Sigma_1 \in \R^{p \times p}\text{ and } \Sigma_2 \in \R^{q \times q}
\]
should not only be reliable, but it should also respect the computational motivation for separability itself. If separability is being considered because unrestricted covariance estimation is infeasible, then a test whose implementation requires unrestricted covariance estimation may be of limited practical value. Ideally, the computational burden of testing separability should not overshadow the cost of manipulating a separable covariance. In particular, when separability potentially holds, one would like to test the assumption at a complexity comparable to that of fitting the separable model, which is essentially the complexity of matrix multiplication of two matrix-valued samples.

The classical likelihood ratio test (LRT) for covariance separability has long been the standard benchmark in low-dimensional Gaussian settings \citep{lu2005likelihood,MITCHELL2006}. However, its direct use is incompatible with the computational principle outlined above, since any LRT requires fitting the model under the alternative, which amounts here to unstructured covariance estimation. From the statistical perspective, the LRT loses power fast with increasing dimension when properly calibrated, because of the variability associated with the unstructured estimate. For these reasons alone, the LRT is not suitable in other than the low-dimensional regime, albeit being the most powerful and natural option in the fixed $p$ and $q$ regime.

On the other hand, several important contributions to separability testing were recently made in the infinite-dimensional regime \citep{aston2017tests, lynch2018,vanDelft2024general,dette2025testing}. These methods address settings in which the observations are (possibly discretized versions of) smooth, functional data. Despite being highly relevant to modern data analysis, their strengths are not identical to those of the classical LRT. In a regular regime, where the LRT is computable, the projection-based or functional tests typically achieve substantially lower power.

Between these two extremes lies another asymptotic regime, the high-dimensional one, where $p$ and $q$ are large enough so that the unrestricted covariance estimation is ill-posed, but the data are still modeled as finite-dimensional matrices rather than as observations from an infinite-dimensional stochastic process, and the notion of likelihood itself is well-defined. 
However, despite its existence, in this setting, the (unstructured) sample covariance becomes ill-conditioned (or singular) and the LRT statistic is often no longer computable or its null distribution departs from the classical asymptotics, and its power degrades accordingly.
Our approach targets precisely this regime, avoiding both limitations: it does not rely on the LRT's well-conditioning assumptions, nor does it restrict attention to a fixed low-dimensional set of directions.

Our starting point is the simple observation that, within the elliptical family, separability can be recast as sphericity after separable whitening. If $\Sigma_1$ and $\Sigma_2$ are two separable covariance factors, then the transformed covariance is the identity exactly when the original covariance is separable. Thus we estimate the separable covariance, whiten the matrix observations by the estimated factors, and test whether the resulting sample is spherical. This converts the separability testing problem into a white-noise or sphericity testing problem, which is familiar in the high-dimensional regime \citep{ledoit2002some,bai2009corrections,chen2010tests}. The matrix whitening step also leads to an interesting unification: several of the most classical tests for sphericity or identity covariance become the same statistic after whitening. In particular the statistic we propose may be viewed equivalently as the matrix-whitened version of John's test \citep{john1971some}, Nagao's test \citep{nagao1973some}, or the Ledoit-Wolf test \citep{ledoit2002some}. The unified statistic can be calculated in three simple steps:
\begin{enumerate}
    \item estimate the separable covariance (under the null),
    \item whiten the data using the estimated row and column factors, and
    \item calculate the Frobenius departure of the empirical covariance of the whitened data from the appropriately scaled identity.
\end{enumerate}
Importantly, the last step can be done efficiently on the level of whitened data, without the need to ever calculate the empirical covariance.

The null distribution of the resulting statistic is not available in a simple closed form in the high-dimensional regime considered here, primarily because the whitening factors need to be estimated from the data. We therefore calibrate the test by Monte Carlo simulation. Under a correctly specified distributional model, this yields finite-sample level control conditional on the estimated separable covariance, and the invariance of the construction makes the calibration independent of the unknown separable factors. Since our test can be seen as a high-dimensional generalization of the Gaussian LRT, Gaussian calibration is a natural first choice.

However, the simulations in this paper show that the effect of the whitening step is much more delicate outside the Gaussian family. For heavy-tailed matrix-variate distributions with the small to moderate sample sizes, the variability introduced by estimating the separable whitening factors can be severely underestimated by Gaussian Monte Carlo calibration. This can lead to inflated rejection probabilities even when the covariance is separable. In other words, a naive Gaussian version of the test may reject not because the covariance is non-separable, but because the data-generating distribution is heavier-tailed than assumed.

To address this issue, we introduce an angular version of the test. After separable whitening, each observation is projected onto the unit sphere by radial normalization. The test statistic is then computed from these normalized directions rather than from the whitened observations themselves. The additional spherical projection removes radial information and therefore reduces sensitivity to the tail behavior of the underlying elliptical distribution.
The resulting angular test is designed to preserve the covariance-shape information relevant for separability while counterbalancing the excess variability caused by heavy-tailed radial components.

Exact Monte Carlo level control can only be guaranteed when the data-generating distribution used in the calibration is known. Nevertheless, our simulation study demonstrates that the angular version is substantially more robust to distributional misspecification than the non-angular version. Across a broad range of dimensions, sample sizes, covariance structures, and elliptical distributions, the angular test achieves reasonable level control under Gaussian calibration even for heavy-tailed data, while retaining almost all of the power of the non-angular, elliptical test. In the high-dimensional regime, this test clearly outperforms both the low-dimensional LRT of \citet{MITCHELL2006} and the infinite-dimensional test of \citet{aston2017tests}, which is the prevailing method among the functional data tests.

Finally, we prove that the proposed test is consistent against dense alternatives
\[
\frac{1}{\sqrt{pq}}\| \Sigma - \Sigma_1 \ot \Sigma_2\|_F \geq \Delta_0 > 0
\]
for some departure $\Delta_0$, any $N$ large enough and proportional to $pq$, and all $\Sigma_1 \in \R^{p \times p}$ and $\Sigma_2 \in \R^{q \times q}$.
In other words, the test is consistent against departures from separability that are spread throughout the covariance spectrum. This is the natural regime for the Frobenius-norm-based statistic invoking John's, Nagao's, and Ledoit-Wolf's tests. Our results show that after separable whitening, such dense departures from separability remain detectable despite the estimation error induced by the whitening step, provided the estimator is good enough.%

The theoretical considerations are complemented by an extensive simulation study and an application to real acoustic phonetic data. We revisit speech recordings of digits in five Romance languages \citep{pigoli2014distances}, represented both through log-spectrograms and through mel-frequency cepstral coefficient matrices. The data illustrate a situation where originally one-dimensional waveforms are transformed by standard preprocessing steps into matrix-valued observations whose covariance structure is scientifically meaningful but too large to be estimated in an unstructured way \citep{pigoli2018statistical}. Applying the proposed angular test, we find strong evidence against separability across languages and across both representations. The agreement between the log-spectrogram and MFCC analyses suggests that the detected non-separability is not merely an artifact of a particular preprocessing.

During the writing/preparation of this manuscript, an independent preprint of a recent work of \citet{sung2026testing} appeared, addressing the same problem of testing separability in the high-dimensional regime. The preprint is based on the same initial idea of recasting separability testing as sphericity testing after whitening. This is not particularly surprising, given the recent increase of interest in separable covariance, as documented by the uptick of publications on the topic (see the numerous references above). Moreover, while the testing problem has been long solved in the low-dimensional regime by the LRT test \citep{lu2005likelihood,MITCHELL2006}, and the functional regime was successfully addressed in recent years as well \citep{aston2017tests, dette2025testing}, testing in the high-dimensional regime has remained unexplored until now. To the best of our knowledge, the ideas underlying the present work and that of \citet{sung2026testing} were developed independently, and the overlap reflects the naturalness of the whitening viewpoint in the high-dimensional separability testing problem, or even in other structural covariance assumption testing contexts. The present paper, however, differs in several respects. Firstly, we build our test from first principles, showing it is equivalent to the well-known sphericity tests after matrix whitening. Secondly, we introduce the angular version of the test, which is crucial for robustness under elliptical heavy-tailed distributions and distributional mis-specification. Thirdly, we establish consistency in a dense high-dimensional alternative regime. Finally, our simulations and real data application emphasize the practical consequences of the whitening variability and the stabilizing effect of spherical projection.

The rest of the paper is organized as follows. Section 2 introduces notation and background concepts. Section 3 defines the matrix-whitened John/Nagao/Ledoit-Wolf statistic, describes the Monte Carlo calibration procedure, and establishes the consistency of the elliptical test. After developing the test as originally intended, Section 4 provides a detailed comparison to the related work of \citet{sung2026testing}. Section 5 introduces the angular version of the test, which is the version of the test that we recommend using. All proofs of the theoretical results are deferred to the appendix. Section 6 underlines the usefulness of the proposed test by empirical demonstrations. The paper concludes with a short discussion.

\section{Preliminaries}\label{sec:background}

For a matrix variate distribution \citep{gupta2018}, the mean and the covariance are naturally defined entry-wise. For the sake of exposition, we will assume that the mean is zero and will be interested in the covariance. When the mean is unknown, computations are to be based on the centered sample.

We denote by $\|\cdot\|_F$ the Frobenius norm, whose square is the sum of squared entries of a given array, while $\|\cdot\|_2$ denotes the spectral norm, i.e.~the largest singular value of a given matrix. For $M \in \R^{d \times d}$ positive definite, $M^{1/2}$ denotes the unique positive definite square-root. $I_{d}$ denotes the identity matrix of size $d \times d$ while $I_{p,q}$ denotes the identity operator on $\R^{p\times q}$. The matrix dimensions $p$ and $q$, and hence all covariance-related objects, depend on the sample size $N$ throughout asymptotic arguments. However, this dependence is usually suppressed in the notation for the sake of conciseness.

For a random matrix $X \in \R^{p \times q}$ with finite second moments and $\E X = 0 \in \R^{p \times q}$, let $\Sigma_{ijkl} = \cov(X_{ij},X_{kl})$. Then $\Sigma = (\Sigma_{ijkl})_{i,j,k,l=1}^{p,q,p,q} \in \R^{p \times q \times p \times q}$ is the covariance tensor. Using the inner product $\langle A , B \rangle := \tr(A^\top B)$ for $A,B \in \R^{p\times q}$, the covariance $\Sigma$ of a random matrix $X \in \R^{p\times q}$ can be equivalently defined as the linear operator on $\R^{p\times q}$ satisfying
\[
\Sigma(A)= \E[\langle X,A \rangle_F X],\qquad A \in \R^{p \times q}.
\]
As such, the covariance operator is self-adjoint and positive semi-definite. Since we work within the framework of maximum likelihood, we will assume strict positive definiteness throughout the paper. Using the tensor product notation \citep{weidmann2012}, the covariance is defined by $\Sigma  =\E[ X\otimes X]$.

\newpage
The covariance is separable if
\[
\Sigma_{i,j,k,l} = (\Sigma_1)_{i,k}\,(\Sigma_2)_{j,l}
\]
for some $\Sigma_1 \in \R^{p \times p}$ and $\Sigma_2 \in \R^{q \times q}$.
In the coordinate-free tensor product notation, the covariance is separable if
\begin{equation}\label{eq:separability}
  \Sigma = \Sigma_1 \ot \Sigma_2  
\end{equation}
for some $\Sigma_1 \in \R^{p \times p}$ and $\Sigma_2 \in \R^{q \times q}$, where $\ot$ is again the tensor product defined by its action on rank one elements as
\[
(\Sigma_1 \ot \Sigma_2)(x \otimes y) = \Sigma_1 x \otimes \Sigma_2 y, \qquad x \in \R^p, y \in \R^q.
\]
Notice that the difference between $\ot$ and $\otimes$ amounts to a simple permutation of the dimensions: $\Sigma_1 \otimes \Sigma_2 \in \R^{p \times p \times q \times q}$, while $\Sigma_1 \ot \Sigma_2 \in \R^{p \times q \times p \times q}$.

\begin{definition}\label{def:equivariance}
    An estimator $\mmle$ of a separable covariance $\Sigma_1\ot\Sigma_2$ based on a random sample $X_1,\ldots,X_N \in \R^{p \times q}$ is called \emph{matrix affine equivariant} if, for any pair of invertible matrices $A \in \mathbb{R}^{p \times p}$ and $B \in \mathbb{R}^{q \times q}$ it holds
\[
\widetilde{\Sigma}_1 \ot \widetilde{\Sigma}_2 = (A \ot B) (\widehat{\Sigma}_1 \ot \widehat{\Sigma}_2) (A \ot B)^\top,
\]
where $\widetilde{\Sigma}_1 \ot \widetilde{\Sigma}_2$ is the estimator based on the transformed random sample $\widetilde X_1,\ldots,\widetilde X_N \in \R^{p \times q}$ with $\widetilde X_n := (A \ot B)X_n$ for $n=1,\ldots,N$.
\end{definition}

Possibly the most prominent member of the elliptical class is the matrix Gaussian distribution. A random matrix $X \in \R^{p \times q}$ is said to be matrix-variate Gaussian if its law on the space of matrices $\R^{p \times q}$ is Gaussian and its covariance is separable. It follows easily that if $Z \sim \mathcal{N}(0, I_{p,q})$, where the dimensions are clear from the dimensions of the covariance, then we have $(\Sigma_1 \ot \Sigma_2)^{1/2} Z \sim \mathcal{N}(0, \Sigma_1 \ot \Sigma_2)$.

\begin{remark}
We recognize that the exposition above is typically carried out using vectorizations \citep[see e.g.][]{gupta2018} and the Kronecker product notation \citep{van2000ubiquitous}, which are related to the tensor product exposition above via a single formula:
\[
\vec((\Sigma_1 \ot \Sigma_2)X) = \vec(\Sigma_1 X \Sigma_2^\top) = (\Sigma_2 \otimes_{K} \Sigma_1)x,
\]
where $\otimes_{K}$ is the Kronecker product and $x = \vec(X) \in \R^{pq}$ is the vectorization of the matrix $X \in \R^{p \times q}$ (the vectorization operator $\vec(\cdot)$ stacks all columns of $X$ into a long vector). While the tensor product notation requires some initial elucidation, it greatly simplifies the subsequent development. While we will still occasionally need to vectorize in the proofs postponed to the appendix, we avoid the clunky rearrangement operators \citep{van1993approximation} and particular matricizations of higher-order tensors completely.
\end{remark}

Given $N$ i.i.d.\ observations $X_1,\ldots,X_N \sim \mathcal{N}(0,\Sigma_1 \tilde\otimes \Sigma_2)$, the (Gaussian) maximum likelihood estimators (MMLEs) of $\Sigma_1$ and $\Sigma_2$ satisfy the fixed-point equations
\begin{equation}\label{eq:mmle}
\hat\Sigma_1 = \frac{1}{Nq}\sum_{i=1}^N X_i \hat\Sigma_2^{-1} X_i^\top\qquad\&\qquad \hat\Sigma_2 = \frac{1}{Np}\sum_{i=1}^N X_i^\top \hat\Sigma_1^{-1} X_i,
\end{equation}
which are solved via the \emph{flip-flop} iterative algorithm \citep{dutilleul1999mle, srivastava2008}. The equations are identifiable only up to a scalar: one typically imposes a normalization such as $\tr(\hat\Sigma_1)=p$ or $\det(\hat\Sigma_1)=1$. Under separability, within the elliptical family, MMLEs converge, up to a scale, to the true covariance parameters~\citep{MMCD}.

If the separability assumption~\eqref{eq:separability}
fails, the appropriately scaled version of these estimators no longer converge to the true covariance components. They instead converge to the \emph{pseudo-MLEs} $(\Sigma_1,\Sigma_2)$, i.e., the minimizers of the Kullback-Leibler divergence from the true (unstructured) distribution within the matrix Gaussian family \citep{white1982maximum}, defined implicitly via
\begin{equation}\label{eq:pseudo-parameter}
\Sigma_1=\frac{1}{q}\mathbb{E}(X\Sigma_2^{-1}X^\top)\qquad\&\qquad \Sigma_2=\frac{1}{p}\mathbb{E}(X^\top\Sigma_1^{-1}X).
\end{equation}

\section{The Elliptical Test}\label{sec:test}

Consider the following generative model for the data.

\smallskip
\noindent\textbf{Model 1.} For $n=1,\ldots,N$, let $X_n = \Sigma^{1/2}Z_n$, where $Z_n$ are i.i.d., following an absolutely continuous, orthogonally invariant distribution $F_0$ such that $\E Z_n=0$ and $\cov(Z_n) =c I_{p,q}$, for $c>0$. For the sake of identifiability of the model, it is often assumed that the scale $c$ is absorbed into the covariance $\Sigma$. For the simplicity of the notation, we adopt that convention and hereafter take $c=1$.
\smallskip

For the given random sample of matrices $X_1,\ldots,X_N \in \R^{p \times q}$ with mean $\mu$ and covariance $\Sigma \in \R^{p \times q \times p \times q}$, we aim to test the hypothesis
\[
H_0: \Sigma = \Sigma_1 \ot \Sigma_2 \text{ for some } \Sigma_1\in \R^{p \times p}, \Sigma_2\in \R^{q \times q}
\]
against $H_1: \neg H_0$. We will use the following statistic.

\begin{definition}\label{def:stats}
    For a random sample $X_1,\ldots,X_N \in \R^{p \times q}$, we define the John's statistic $T_N^{(Jo)}$, the Nagao's statistic $T_N^{(Na)}$ and the Ledoit-Wolf statistic $T_N^{(LW)}$ as
    \[
    \begin{split}
    T_N^{(Jo)}(X_1,\ldots,X_N) &= \frac{1}{pq} \left\| \frac{pq\widehat{C}}{\tr(\widehat C)}-I_{p,q}\right\|_F^2, \\
    T_N^{(Na)}(X_1,\ldots,X_N) &= \frac{1}{pq} \left\| \widehat{C}-I_{p,q}\right\|_F^2, \text{ and}\\
    T_N^{(LW)}(X_1,\ldots,X_N) &= \frac{1}{pq}\|\widehat{C}-I_{p,q}\|_F^2
-\frac{pq}{N}\left(\frac{1}{pq}\operatorname{tr}(\widehat{C})\right)^2,
    \end{split}
    \]
    where $\widehat{C}=\frac{1}{N} \sum_{n=1}^N Y_n \otimes Y_n$ is the empirical covariance estimator of the separably whitened sample, i.e.~$Y_n = (\mmle)^{-1/2} X_n$ with $(\widehat{\Sigma}_1,\widehat{\Sigma}_2)$ being the MMLE \eqref{eq:mmle}.
\end{definition}

Let $T_N$ be any of the three statistics above. The test rejects the null hypothesis if the p-value calculated via the following Monte Carlo (MC) simulation scheme, using $M\in\N$ MC samples, is smaller than or equal to the given level $\alpha \in (0,1)$:

\newpage
\begin{itemize}
    \item For $m=1,\ldots,M$:
\begin{enumerate}
    \item Draw $Z_1^\star,\ldots,Z_N^\star \in \R^{p \times q}$ independently from $F_0$.
    \item Calculate and store the statistic $T_{N,m}^\star = T_N(Z_1^\star,\ldots,Z_N^\star)$.
\end{enumerate}
    \item Return $\widehat{p}_{N,M}=\big(\# [T_{N,m}^\star \geq T_N]+1\big)\big/\big(M+1\big)$
\end{itemize}
For example, if the underlying distribution $F_0$ is Gaussian, entries of the matrices generated in Step 1 above are sampled as independent $\mathcal{N}(0,1)$-distributed variables. We recommend $M=999$, as is common in Monte Carlo testing \citep{davison1997bootstrap}. 

Notice that the whitening step in Definition~\ref{def:stats} as well as  Step 2 of the for-loop involve separable covariance estimation and a subsequent whitening step. For that, we always utilize \eqref{eq:mmle} as our covariance estimator. Appendix~\ref{sec:covariance_map} justifies this choice by showing that, on the population level and under $H_0$, any specific estimation strategy for obtaining a separable covariance that is matrix affine equivariant always leads to the right separable covariance target up to a proportionality constant. The constant is immaterial for our test as it disappears during the matrix whitening step.

There are two more important observations to be made about the proposed testing scheme. Firstly, the three statistics coincide (cf.~Lemma \ref{lem:statistics_equal} below). This is in contrast to general sphericity testing. There, to begin with, John's statistic is scale invariant while Nagao's and Ledoit-Wolf statistics are not. Specifically, John's statistic is designed to test whether the covariance is proportional to identity, while Nagao's and Ledoit-Wolf statistics test whether the covariance is equal to identity. Moreover, the Ledoit-Wolf statistic is asymptotically unbiased even in the high-dimensional regime \citep{wang2013sphericity,li2018structure}, while John's and Nagao's statistics are not. In our setup, however, $T_N^{(Jo)}$ and $T_N^{(Na)}$ coincide exactly as a result of the separable whitening, which removes the scale. $T_N^{(LW)}$ then differs by an additive constant only, which is automatically adjusted for by the Monte Carlo scheme. Overall, we propose only a single test since, interestingly, the three of the most popular sphericity tests coincide in our setup. The second observation is that the test statistics can be calculated efficiently, without the need to ever store, or even calculate, the covariance tensor.

\begin{lemma}\label{lem:statistics_equal}
    It holds $\tr(\widehat{C})=pq$ and hence $T_N:=T_N^{(Jo)}=T_N^{(Na)}=T_N^{(LW)}+\frac{pq}{N}$. Furthermore, $T_N=(N^2pq)^{-1} \sum_{n=1}^N \sum_{n'=1}^N \langle Y_n, Y_{n'} \rangle^2 - 1$.
\end{lemma}

As a result, calculating the statistic can be done without forming or storing the empirical covariance in the time complexity $\mathcal{O}\big(Npq\min(N,pq)\big)$. In the high-dimensional regime  $N=O(pq)$, this reduces to $\mathcal{O}(N^3)$, which dominates the number of operations needed for one step of the flip-flop algorithm for calculating the MMLE $\mathcal{O}\big(Npq(p+q)+p^3+q^3\big)$. Hence, the overall complexity of the test is cubic in sample size $N$.

\newpage
\subsection{Exact Level Control}

The MC calibration introduced in the previous section yields exact level control of the proposed test (up to the approximation error stemming from the finite number of MC samples $M$), as long as a suitable estimator of the separable covariance is used when calculating the statistic. The following theorem provides a general invariance statement, which is specialized to the specific statistic in the subsequent corollary.

\begin{theorem}\label{thm:law}
Under the null hypothesis of separability, let $\widehat{\Sigma}_1 \ot \widehat{\Sigma}_2$ be a separable covariance estimator based on $X_1,\ldots,X_N$, with $X_n=\Sigma^{1/2}Z_n$,
where $Z_n$ are i.i.d.from an absolutely continuous distribution~with $\E Z_n=0$ and $\cov(Z_n) = I_{p,q}$, $n=1,\ldots,N$, such that
\begin{enumerate}
    \item it is uniquely defined for the given data set, and
    \item it is matrix affine equivariant in the sense of Definition~\ref{def:equivariance}.
\end{enumerate}
Let $T_N$ be any orthogonally invariant statistic in the sense that $T_N(X)\sim T_N(UXV^\top)$, for any orthogonal $U\in\R^{p\times p}$ and $V\in\R^{q\times q}$.  
Then $T_N(\widetilde X_1,\ldots,\widetilde X_N) \sim T_N(\widetilde Z_1^\star,\ldots,\widetilde Z_N^\star)$, where $\widetilde X_n=(\mmle)^{-1/2} X_n$ and $\widetilde Z_n^\star=(\widetilde\Sigma_1\ot\widetilde\Sigma_2)^{-1/2} Z_n^\star$, with $\widetilde\Sigma_1\ot\widetilde\Sigma_2$ the separable covariance estimator based on $Z_1^\star,\ldots,Z_N^\star$ independent copies of $Z_1,\ldots,Z_N$.
\end{theorem}

\begin{corollary}\label{cor:cor_of_thm2}
     Under the assumptions of the previous theorem, let $T_N$ be any statistic that depends on the sample $X_1,\ldots,X_N$ only through the eigenvalues of
$\widehat{C}$.\\ Then $T_N(X_1,\ldots,X_N) \sim T_N(Z_1^\star,\ldots,Z_N^\star)$, where $Z_1^\star,\ldots,Z_N^\star$ are independent copies of $Z_1,\ldots,Z_N$.
\end{corollary}

Corollary~\ref{cor:cor_of_thm2} states that, under $H_0$, any of the three test statistic considered in Definition~\ref{def:stats} has a distribution that does not depend on the unknown separable covariance
as long as an estimator that is unique and matrix affine equivariant is used for the whitening step. The MMLE satisfies the conditions 1.~and 2.~in the previous theorem as showcased in the following proposition, leading to exact level control of the proposed test, stated as Corollary~\ref{cor:exact_level}.

\begin{proposition}\label{prop:MMLE}
    Under the assumptions of Theorem~\ref{thm:law}, the MMLE satisfies the conditions of Theorem~\ref{thm:law} almost surely under the sample size condition $N \geq \lfloor p/q+q/p \rfloor  + 2$. 
    On the other hand, the partial tracing estimator of \citet{aston2017tests} or the best separable approximation \citep{genton2007separable,masak2023separable} do not satisfy the second condition.
\end{proposition}

The sample size condition is for the MMLE to uniquely exist almost surely \citep[it is a sufficient, and almost necessary condition -- almost due to edge cases, cf.][]{drton2021}.
For this, neither distributional assumptions on $X$ nor  separability of the ground truth $\Sigma$ are needed.

\begin{corollary}\label{cor:exact_level}
Let $N \geq \lfloor p/q+q/p \rfloor  + 2$. Consider the p-value $\widehat{p}_{N,M}$ computed via the Monte Carlo scheme in the previous section. Under the assumptions of Theorem~\ref{thm:law}, for $\alpha \in (0,1)$, the test that rejects if and only if $\widehat{p}_{N,M} \leq \alpha$ satisfies $P(\widehat{p}_{N,M} \leq \alpha) = \lfloor \alpha (M+1) \rfloor / (M+1) \leq \alpha$.
\end{corollary}

The test is thus an exactly valid finite-sample level-$\alpha$ test up to the Monte Carlo approximation error.

\subsection{High-dimensional Consistency}

Having established level control under $H_0$, the natural next step is to show that the power of the test tends to one under $H_1$. We will show consistency in the high-dimensional regime, where $p$ and $q$ increase with $N$, under the following assumptions.
\begin{description}
    \item[(A1)] Let $p/\sqrt{N} \to \gamma_1$ and $q/\sqrt{N} \to \gamma_2$ with $\gamma_1,\gamma_2 \in (0,\infty)$.
    \item[(A2)] Let $\|\Sigma\|_\op$, $\|\Sigma_1\|_\op$ and $\|\Sigma_2\|_\op$ be bounded, and let $\lambda_{\text{min}}(\Sigma_1)$ and $\lambda_{\text{min}}(\Sigma_2)$ be bounded away from zero, all uniformly in $N$ that is large enough.
    \item[(A3)] Let $X_n = \Sigma^{1/2}Z_n$, $n=1,\ldots,N$, where $Z_n$ are i.i.d.~with mutually independent entries $\E Z_n=0$, $\cov(Z_n)=I_{p,q}$, and finite eighth moments.
    \item[(A4)] Let $\frac{1}{\sqrt{pq}}\|\widehat{\Sigma}_1 \ot\widehat{\Sigma}_2 - \Sigma_1 \ot \Sigma_2 \|_F = o_p(1)$, where $(\Sigma_1,\Sigma_2)$ is the pseudo-true parameter defined in~\eqref{eq:pseudo-parameter}.
    \item[(A5)] Let $\frac{1}{\sqrt{pq}}\|\Sigma - \Sigma_1 \ot \Sigma_2 \|_F \geq \Delta_0 > 0$ for all $N$ large enough.
\end{description}
Firstly, let us discuss the individual assumptions. Assumptions (A1)-(A3) are fairly standard. Notice (A1) imposes the high-dimensional regime, where $pq \asymp N$. However, separating the convergence of the row and column dimensions ensures we are not in the ultra-rectangular case, where $p \gg q$ or vice versa.

On the other hand, Assumption (A4) is problem-specific. Under $H_0$ and Gaussianity, the much stronger rate $\|\widehat{\Sigma}_1\ot\widehat{\Sigma}_2 - \Sigma_1 \ot \Sigma_2 \|_F= \mathcal{O}_p(1)$ holds \citep{franks2026near}. Nonetheless, the rate of convergence of the pseudo-true parameter under $H_1$ has not been explored yet, to the best of our knowledge. The underlying proof strategy in \citet{franks2026near} depends crucially on the invariance of the separable class to congruences of the same class, and hence it is difficult to generalize their analysis to the pseudo-true parameter. In general, the slightly weaker rate $\frac{1}{\sqrt{pq}}\|\widehat{\Sigma_1}\ot\widehat{\Sigma}_2 - \Sigma_1 \ot \Sigma_2 \|_F=\mathcal{O}_p(1)$ is available for an unstructured covariance, and thus the standard MLE misspecification theory \citep{white1982maximum} can be used to obtain the rate. It is reasonable to expect that identifying the pseudo-true separable parameter in the non-separable case can be done in a slightly better, more informed way than unstructured covariance estimation. However, proving this is beyond the scope of the present paper.

Finally, Assumption (A5) guarantees high-dimensional departure from $H_0$. We would argue that this is very natural, given the aim to test whether $\Sigma=\Sigma_1 \ot \Sigma_2$, or equivalently $\|\Sigma-\Sigma_1 \ot \Sigma_2\| > 0$. It is part of the proof of the following theorem (cf.~the Appendix) that (A5) can be translated to a dense departure from sphericity on the whitened level.

\begin{theorem}\label{thm:consistency}
    Under Assumptions (A1)-(A5), the Monte Carlo test is consistent, i.e.~for a given level $\alpha \in [1/(M+1),1]$ it holds $P(\widehat{p}_{N,M} \leq \alpha) \stackrel{p}{\to} 1$ for $N \to \infty$.
\end{theorem}

\section{Connections to Prior Art}\label{sec:art}

Our tests estimate the separable covariance under the null exactly as the LRT of \citet{MITCHELL2006} does.
However, our test avoids the computationally prohibitive and statistically unsound estimation under the alternative. From another perspective, the LR statistic can be written as
\[
-2 \log\Lambda = N\sum_{j=1}^{pq}(\widehat\lambda_j - \log \widehat\lambda_j - 1)
\]
where $\{\widehat\lambda_j\}$ are eigenvalues of the matrix-whitened data. It is thus clear that the statistic is not well defined unless $N \geq pq$. Moreover, this spectral statistic is not linear, which impedes efficient computation.

On the other hand, the proposed test statistic can also be understood as a measure of dispersion of the eigenvalues of the whitened empirical covariance $\widehat C$. Under separability, it should be $\widehat{C} \approx I_{p,q}$, so the eigenvalues $\lambda_1,\ldots,\lambda_{pq}$ should concentrate around one. Thus, any non-separable structure that whitening cannot remove is spread out. A natural summary of this spread is the standard deviation of the spectrum. Because $\tr{\widehat{C}}=pq$ after whitening (see Lemma~\ref{lem:statistics_equal}), the eigenvalues have mean one, and hence
\[
T_N^{(Na)}=\frac{1}{pq}\sum_{j=1}^{pq}(\lambda_j-1)^2=\frac{1}{pq}\sum_{j=1}^{pq}(\lambda_j-\bar\lambda)^2.
\]
In words, the statistic is exactly the variance of the whitened eigenvalue spectrum. The intuitive dispersion measure thus coincides with the John/Nagao/Ledoit-Wolf statistic. Moreover, the relationship to the LRT test can be understood through the approximation $\lambda-\log\lambda-1=\tfrac12(\lambda-1)^2+o((\lambda-1)^2)$ near $\lambda=1$. The proposed statistic is thus the second-order surrogate of the LRT's statistic, specifically $-2\log\Lambda\approx \tfrac{Npq}{2}T_N$. By this approximation, the logarithmic singularity at zero is avoided, leaving $T_N$ well defined even when $N < pq$, which is when the LRT breaks down. Secondly, the spectral variance makes the statistic a dense detector, sensitive to departures spread across many eigenvalues, the regime formalized in Assumption~(A5). Finally, the spectral viewpoint also showcases how whitening removes the global scale, which is why the scale-invariant John statistic and the scale-dependent Nagao statistic coincide here. Thus the eigenvalue dispersion is a pure measure of the shape of the whitened cloud. This raises the question of whether one can remove the scale even more aggressively, at the level of each individual observation, to insulate the test from nuisance features of the sampling distribution. We pursue this in Section~\ref{sec:angular}, where normalizing each whitened observation to the unit sphere yields the angular test.

The test of separability that is most commonly adopted for functional data \citep{aston2017tests} uses partial tracing to estimate the separable covariance under the null, i.e.~it does not take advantage of the maximum likelihood geometry. Thus the test works even in the infinite-dimensional regime, where the maximum likelihood geometry is ill-posed, but lacks power in regular regimes, where the partial tracing estimators are inferior to MMLE. As a result, the test of \citet{aston2017tests} cannot compete with the proposed test in the high-dimensional regime (cf.~Appendix~\ref{sec:sim_aston}), which is why we do not compare our method with the one of \citet{aston2017tests} in the main simulation study in Section~\ref{sec:empirical}. Another reason for this is that the test of \citet{aston2017tests} depends on the choice of two tuning parameters (row and column ranks under separability), and no automatic selection procedure for these is available.

Interestingly, our statistic, e.g.~Nagao's one, can be written as
\[
T_N^{(Na)} = \frac{\|\widehat{C}\|_F^2}{pq} - 1.
\]
This is very close to the following statistic proposed by \citet{sung2026testing}:
\[
T_3 = \frac{\|\widehat{C}\|_F^2}{\widehat{\sigma}_1^2} - 1,
\]
where $\widehat\sigma_1$ is the leading separable component score of $\widehat{C}$, introduced by \citet{masak2023separable}. Only in the latest version of their manuscript, \citet{sung2026testing} correctly pointed out that their $T_3$ does not in fact depend on the separable component score, which is always constant after matrix whitening, and hence $T_3$ actually coincides with the elliptical test of Section~\ref{sec:test}. Still, our work differs in several important aspects:
\begin{enumerate}
    \item We show consistency for dense alternatives, while \citet{sung2026testing} only study consistency for one of their statistics (not the one that coincides with our elliptical test) under the sparse (partial isotropy) alternative regime.
    \item Our exposition is developed from first principles, relying on the MMLE and its invariance, as well as high-dimensional convergence rates, rather than on more elaborate frameworks such as the Kronecker-core decomposition \citep{hoff2023core} or separable component expansion \citep{masak2023separable}. We believe this makes the genesis of the test as well as the theoretical results more easily accessible.
    \item We propose the angular version of the test (in the next section) to alleviate the necessity to know the underlying distribution exactly. Consistency is shown again under the dense alternative regime.
    \item The proposed test is computationally efficient, and distributed in a ready-to-use \texttt{R}~package \texttt{elseptest} with pre-tabulated quantiles, making our method appealing in practice.
\end{enumerate}

Finally, our assumptions to prove consistency of the test are different, in some sense weaker than those of \citet{sung2026testing}. In particular, the rate in Assumption (A4) can be immediately replaced by $\|\widehat{\Sigma}_1 \ot\widehat{\Sigma}_2 - \Sigma_1 \ot \Sigma_2 \|_2 = o_p(1)$, which is closely related to Assumption (A3) of \citet{sung2026testing}. The latter postulates the spectral norm rate of $\mathcal{O}_p\big(\log(N)/\sqrt{N}\big)$, which is naturally stronger than the $o_p(1)$ that would guarantee Assumption (A4) of Section~\ref{sec:test}. However, \citet{sung2026testing} make the assumption on the whitened level directly, without any conditioning assumption such as our (A2). The proof techniques are completely different, and so any closer comparison seems difficult. Secondly, we indeed replace the assumed Frobenius norm rate (A4) for an operator norm assumed rate (B4) for the angular test that will be introduced in the subsequent section, because this test utilizes radial normalization, which is nonlinear, and so more uniform control of the whitening operator is required. Still, the rate we assume below in (B4) is just $o_p(1)$ and thus the previous comparison applies as is. That being said, our own numerical results suggest that the stronger rates assumed by \citet{sung2026testing} are valid, and the authors provide a more detailed asymptotic analysis, albeit only of the spectral procedure and under spiked departures from sphericity on the whitened level, which cannot be easily translated to spiked departures on the original covariance of the data and differs substantially from the dense departure on the original level we work with. Our results should thus be seen as complementary in numerous aspects.

\section{The Angular Test} \label{sec:angular}

In Section~\ref{sec:test}, we showed how to test for separability of the covariance tensor of matrix-variate data by recasting the problem as a sphericity test after MMLE-whitening. The whitening step removes the scale, leading to two of the most well-known sphericity tests (Nagao's test and John's test) to coincide in this case, while the Ledoit-Wolf statistic amounts to a monotone transformation of the other two statistics, which the Monte Carlo calibration is naturally impervious to. In this section, we propose to follow the whitening step by an additional spherical projection step, introducing a new testing procedure.

\begin{definition}\label{def:angular_test}
    For a random sample $X_1,\ldots,X_N \in \R^{p \times q}$, we define the angular statistic
    \[
    T_N^{(A)} = pq \left\| \widehat{S} - \frac{1}{pq} I_{p,q}\right\|_F^2,
    \]
    where $\widehat{S} = \frac{1}{N}\sum_{n=1}^N U_n \otimes U_n$ is the empirical covariance of the matrix-whitened and sphere-projected sample, i.e.
    \[
    U_n = \frac{(\mmle)^{-1/2} X_n}{\| (\mmle)^{-1/2} X_n \|_F}
    \]
    with $(\widehat \Sigma_1,\widehat \Sigma_2)$ being the MMLE.
\end{definition}

The angular statistic is then used within the Monte Carlo scheme of Section~\ref{sec:test} to produce the p-value of the angular test. Note that both centering and the scale of the angular statistic are different compared to the statistics from Definition~\ref{def:stats}. This is needed to obtain high-dimensional consistency, but as a monotone transformation of the statistic, it does not affect the finite sample performance in any way. Still, one might expect that removing the radius component completely amounts to a loss of power. However, this does not seem to be the case: we demonstrate in Section~\ref{sec:empirical} that the empirical power does not suffer from the spherical projection step. From an intuitive point of view, the angular test is the version of the elliptical test, where instead of the sample covariance $\hat{C}$ of the whitened data, one uses the spatial-sign covariance~\citep{DURRE201680} instead. Under an elliptical family, these two, on the population level, share the eigenvectors, but the eigenvalues are not the same (or proportional), except in the case where $\hat{C}$ is isotropic. The eigenvalues of the spatial-sign 
are a distorted version of the true ones, typically pulled closer together (i.e.~shrunk). In that sense, the variability within the eigenvalues is less pronounced for the spatial-signed covariance used in the angular test. 
However, on the population level, the specific shrinkage depends on the spectrum of $\hat{C}$, but not on $F_0$. Thus, the matching transformation also applies on the level of MC, thus potentially explaining why the test does not lose power, as well as why it remains a useful diagnostic despite discarding the radial information.

On the other hand, what we hope to achieve by this is to make the test more robust with respect to the choice of the underlying distribution $F_0$ used for Monte Carlo. In Model 1, $Z_n$ has stochastic representation as $Z_n=r_n V_n$, giving $X_n=r_n\Sigma^{1/2}V_n$. Writing $\tilde X_n=(\hat\Sigma_1\ot\hat\Sigma_2)^{-1/2}X_n$, $\|\tilde X_n\|_F=r_n\|(\hat\Sigma_1\ot\hat\Sigma_2)^{-1/2}\Sigma^{1/2}V_n\|$. Thus, working with $\tilde X_n/\|\tilde X_n\|_F$ as opposed to $\tilde X_n$, removes the effect of the univariate radial component $r_n$ as well as the overall scale.

Thus, ignoring the fact that we use the data to estimate the separable approximation of the covariance, the spherical projection removes the only component of $X_n$ which is distribution specific, its radial part $r_n$. More specifically, if the whitening were performed with the true covariance, under $H_0$, the resulting angular projection $\Sigma^{-1/2} X_n/\|\Sigma^{-1/2} X_n\|_F=U_n$ would have spherical uniform distribution (see e.g. \cite{DURRE201680}), regardless of~$F_0$. 

While this is something difficult to capture in a theoretical analysis (because exact level control is not possible without the knowledge of the underlying distribution, while asymptotic level control would require the study of asymptotic distribution of the test statistic, something that is difficult in the high-dimensional regime even in the case of more straightforward tests), in practical terms, the spherical projection step alleviates the necessity to know the underlying distribution $F_0$ exactly.

\begin{proposition}\label{prop:angular_level}
    Let $N \geq \lfloor p/q+q/p \rfloor  + 2$, and consider the p-value $\widehat{p}_{N,M}$ obtained via the Monte Carlo scheme from Section~\ref{sec:test} with $T_N := T_N^{(A)}$. Under the assumptions of Theorem 1, for $\alpha \in (0,1)$, the test that rejects if and only if $\widehat{p}_{N,M} \leq \alpha$ is a finite-sample size-$\alpha$ test, up to the Monte Carlo approximation error.
\end{proposition}

Secondly, the angular test statistic also has the inner product representation, allowing for an efficient computation.

\begin{lemma}\label{lem:angular_computation}
    The angular test statistic $T_{N}^{(A)}$ of Definition~\ref{def:angular_test} satisfies
        \[
        T_N^{(A)} = \frac{pq}{N^2} \sum_{n=1}^N\sum_{m=1}^N \langle U_n, U_m\rangle^2 -1.
        \]
\end{lemma}

On the other hand, the spherical projection adds another non-linear step to the statistic's calculation procedure, so the theoretical considerations become slightly more cumbersome. In particular, while it is easy to show for the elliptical test (cf.~the proof of Theorem~\ref{thm:consistency}) that a dense departure from separability guarantees a dense departure from sphericity on the matrix-whitened level, showing the same requires additional work when the spherical projection is involved. Hence we provide the following statement separately.

\begin{proposition}\label{prop:relevant_hypothesis}
    Let
	\[
	S(C):=\E\!\left[\frac{Y \otimes Y}{\|Y\|^2}\right] \qquad \&
	\qquad
	Y=C^{1/2}Z,
	\]
	where $Z$ is spherical in \(\R^{p \times q}\), $C\succ0$, and $\tr(C)=pq$.
	Then
	\[
	S=\frac{1}{pq} I_{p,q}
	\quad\Leftrightarrow\quad
	\Sigma=\Sigma_1 \ot \Sigma_2.
	\]
	Moreover, under Assumptions (B2) and (B5) of Theorem~\ref{thm:angular_consistency}, there exists
	$\wideparen{\Delta}_0>0$ such that
	\[
	pq\Big\|S-\frac{1}{pq}I_{p,q}\Big\|_F^2
	\ge \wideparen{\Delta}_0
	\]
	for all $N$ large enough.
\end{proposition}

Finally, we can provide the consistency statement for the angular test.

\begin{theorem}\label{thm:angular_consistency}
    Suppose that the following conditions hold:
    \begin{description}
    \item[(B1)] \(p/\sqrt N\to \gamma_1\) and \(q/\sqrt N\to \gamma_2\) with
	\(\gamma_1,\gamma_2\in(0,\infty)\).
	\item[(B2)] \(\|\Sigma\|_{2}\), \(\|\Sigma_{1}\|_{2}\), and
	\(\|\Sigma_{2}\|_{2}\) are bounded, and
	\(\lambda_{\min}(\Sigma)\), \(\lambda_{\min}(\Sigma_{1})\), and
	\(\lambda_{\min}(\Sigma_{2})\) are bounded away from zero, all uniformly for
	\(N\) large enough.
	\item[(B3)] Model 1 holds with $Z_n$, $n=1,\dots,N$, i.i.d.~from an absolutely continuous orthogonally invariant distribution.
	\item[(B4)] The separable covariance estimator is operator-norm consistent in the sense
	\[
	\|\widehat\Sigma_{1}\ot \widehat\Sigma_{2} -\Sigma_{1}\ot \Sigma_{2}\|_{2}
	=o_p(1).
	\]
	\item[(B5)] Let $\frac{1}{\sqrt{pq}}
	\|\Sigma-\Sigma_{1}\ot \Sigma_{2}\|_F
	\ge \Delta_0 > 0 $
	for all $N$ large enough.
    \end{description}
    Then the Monte Carlo test based on $T_N^{(A)}$ is consistent. That is, for a fixed level $\alpha \in [1/(M+1), 1]$, we have $P(\widehat{p}_{N,M} \leq \alpha) \stackrel{p}{\to} 1$ for $N \to \infty$. 
\end{theorem}

The results stated above provide essentially the same theoretical guarantees for the angular test that are available for the elliptical test from the previous section, albeit under slightly different assumptions. Roughly speaking, one can see by comparing Assumptions (B1)-(B5) of the previous theorem to Assumptions (A1)-(A5) of Section~\ref{sec:test} that the angular test requires better conditioning of the ground truth covariance $\Sigma$, milder moment assumptions, and an MMLE rate of convergence w.r.t.~the operator norm, as opposed to the Frobenius norm. Although the operator norm rate is strictly speaking slightly stronger (the Frobenius norm rate of (A4) follows from Assumption (B4)~of the previous theorem directly), we believe it requires a comparable amount of work to obtain the two rates, as suggested by the analysis of \citet{franks2026near}. We show in Appendix~\ref{sec:appendix_sim} numerically that the rates for the pseudo-parameters seem to be the same as the ones provided by \citet{franks2026near} under separability, which are actually stronger than what Assumptions (A4) or (B4) impose.

Nonetheless, we recognize that the provided theory does not sufficiently distinguish between the elliptical test of Section~\ref{sec:test} and the angular test of this section, apart from the mild difference in assumptions. Therefore, in the following section, we provide numerical evidence that the angular test empirically robustifies the whitened test w.r.t.~misspecification of the data generating distribution, while retaining essentially the same power.

\section{Empirical Demonstration}\label{sec:empirical}

\subsection{Simulation Studies}\label{sec:sim}

Recall that we write the two factor dimensions as $p$ and $q$, so that an observation is a matrix $X\in\mathbb{R}^{p\times q}$ and the covariance is a tensor $\Sigma \in \R^{p \times q \times p \times q}$ . Sample size is indexed by a multiplier $m>0$ through
\begin{equation}
  N \;=\; \max(4,\lceil m\,pq \rceil),
  \label{eq:n-mult}
\end{equation}
so that $m<1$ leads to $N<pq$, while $m\ge 1$ corresponds to the regular regime of $N\geq pq$. This relates to the aspect ratio of \citet{sung2026testing} via $\gamma:=pq/N=1/m$. The lower bound of 4 observations is used to ensure the existence of the MMLEs when $p = q$ \citep{soloveychik2016gaussian}. 
We examine two distinct simulation scenarios: First, we replicate the moderate-spike design from \citet{sung2026testing} for Gaussian data. Next, we conduct our main simulation, which involves both Gaussian and matrix-$t$ distributed data based on a specific non-separable model \citep{masak2023separable}.

Critical values are obtained by Monte Carlo under separability. For each configuration $(p,q,N)$, we draw $999$ data sets of size $N$ from the separable null $\Sigma=I_{pq}$ under the \emph{normal} reference distribution and take the empirical $(1-\alpha)$-quantile of each statistic as its critical value at $\alpha=0.05$.
By matrix affine equivariance, this single calibration is, for each configuration $(p,q,N)$, valid for all separable covariances $\Sigma_1 \ot \Sigma_2$, so the corresponding quantiles can be tabulated once per configuration and reused. This is a major advantage of the proposed Monte Carlo tests since we can compute the quantiles of the test statistic under $H_0$ once for each $(p,q,N)$ and store them, resulting in fast computation times, especially compared to bootstrap-based procedures.

Calibrating the null from Gaussian draws while the data may originate from a different distribution is deliberate: it is the practitioner's default when the sampling law is unknown, and the proposed tests are intended to remain approximately pivotal across sampling distributions.

\subsubsection{Setup 1: Spiked Alternative Shrunk toward the Identity}
\label{sec:sim-reproduce}

In the first setting, we reproduce the Gaussian moderate-spike design of \citet{sung2026testing}, in which a non-separable target $\Sigma$ \citep[denoted $C_2$ in][]{sung2026testing} is shrunk toward the identity as in
\begin{equation}
  \Sigma_{w} \;=\; w\,\Sigma + (1-w)\,I_{p,q} ,
  \qquad w\in[0,1],
  \label{eq:c2w}
\end{equation}
so that larger $w$ yields stronger non-separability and $w=0$ returns the separable null. We extend the original grid to include $w=0$, which allows us to also capture the empirical level of each test in the same figure. Our results in Figure~\ref{fig:hoff_01} reproduce those of \citet{sung2026testing} and since we estimate the mean, the relevant comparison is to their unknown-mean variant. At $w=0$ all tests attain approximately the nominal level. The \emph{Elliptical} test performs as well as the \emph{Likelihood Ratio} test for $N > pq$ and performs better than the test based on the largest eigenvalue. The proposed \emph{Angular} test (Definition~\ref{def:angular_test}) performs as well as the \emph{Elliptical} test.

\begin{figure}[!t]
    \centering
    \includegraphics[width=1\linewidth]{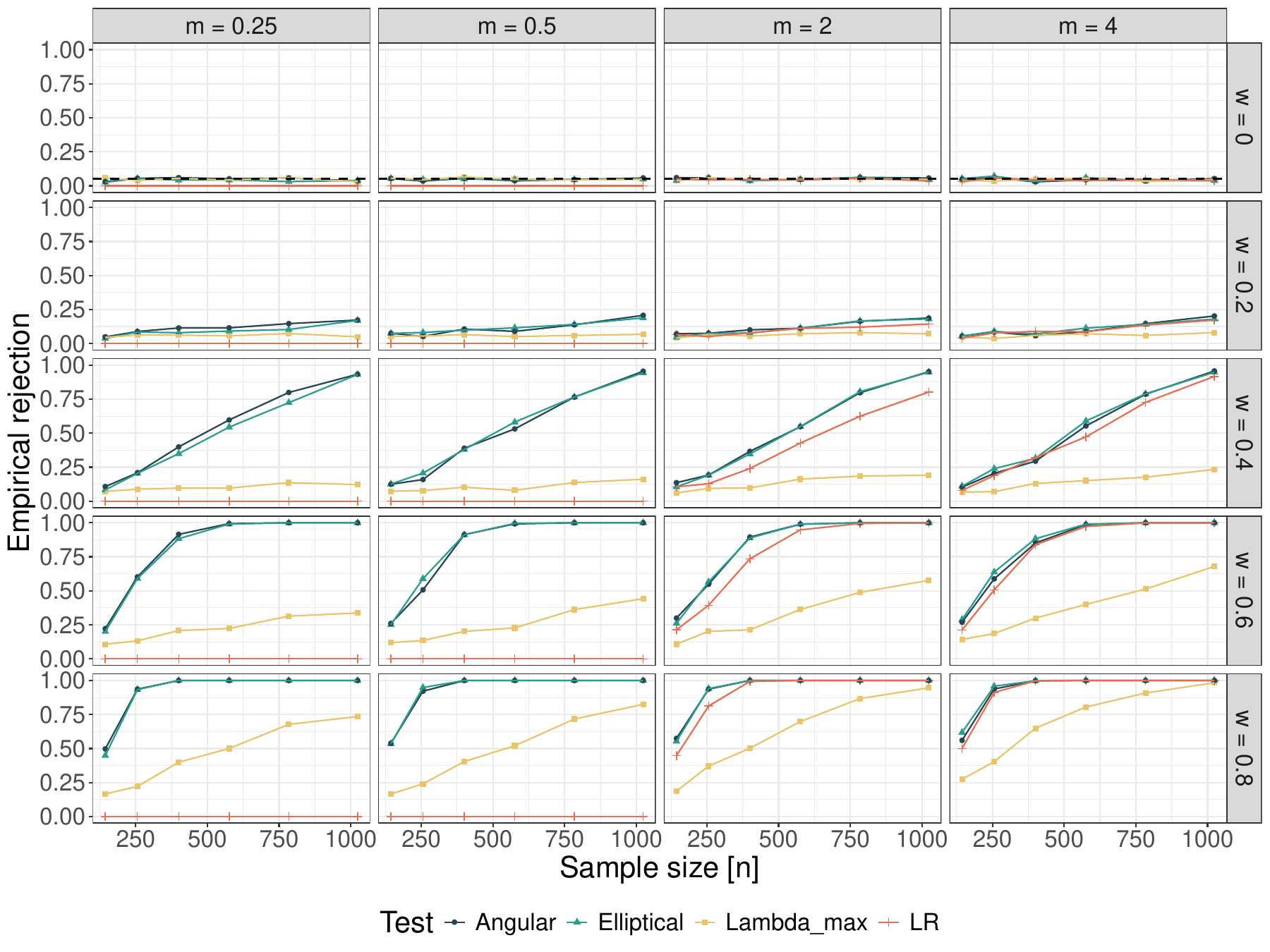}
    \caption{Replication of the results of \citet{sung2026testing} where $m = \gamma^{-1}$, introducing the proposed test as well. \texttt{Lambda\_max} refers to the test based on the largest eigenvalue and corresponds to $\phi_1$ of \citet{sung2026testing} while \texttt{Elliptical} is equivalent to their $\phi_3$.}
    \label{fig:hoff_01}
\end{figure}

\subsubsection{Setup 2: Rank-2 Separable Component Alternative}
\label{sec:sim-family}

Our main synthetic example studies the empirical performance of the proposed tests within a family of targets with a controlled measure of non-separability. Every (non-separable) covariance of matrix-valued data can be decomposed into its \emph{separable components} \citep{masak2023separable} as in
\[
\Sigma = \sum_{r=1}^{pq} \sigma_r A_r \ot B_r,
\]
where $\{A_r\} \subset \R^{p \times p}$ and $\{B_r\} \subset \R^{q \times q}$ are, respectively, orthonormal sets, and $\sigma_1 \geq \sigma_2 \geq \ldots \in \R$ are the \emph{separable component scores}. We generate $\{A_r\}$ and $\{B_r\}$ randomly and fix $\sigma_r=0$ for $r > 2$, obtaining either the separable null (if also $\sigma_2=0$) or the non-separable alternative (if $\sigma_2>0$). We can quantify the departure from separability by the rank-two index
\begin{equation}
  \rho(\Sigma) \;=\; \frac{\sigma_2}{\sqrt{\sigma_1^2+\sigma_2^2}}
          \;\in\;\big[0,\tfrac{1}{\sqrt 2}\big],
  \label{eq:rho}
\end{equation}
which equals to 0 if and only if $\Sigma$ is separable.
For each dimension setting, we generate $1100$ such covariances: $100$ separable ones ($\rho=0$) and $1000$ non-separable ones, with the non-separability index spread to cover the attainable range. The lowest attainable $\rho > 0$ for the non-separable case in our setup increases with increasing dimension $pq$ as a consequence of keeping the conditioning fixed across the dimensions.

For each $\Sigma$, $100$ data sets are drawn, and the rejection decisions are recorded. To display power as a function of $\rho$, the $1100$ covariances, each with $100$ replications, in each panel are grouped into $11$ bins, where the first bin contains the separable targets at $\rho=0$. The per-bin rejection rates are plotted as dots joined by lines. The bin locations differ across dimension settings, reflecting the dimension-dependent range of $\rho$.

Figures~\ref{fig:simulations_setup1} and \ref{fig:simulations_setup1_t} summarize the simulation results and show a grid of panels faceted by dimension across the columns, and by sample-size regime down the rows,
\begin{equation*}
  p=q\in\{5,10,20,50\}, \qquad m\in\{0.1,0.25,0.5,1,2,4\}, \qquad N \;=\; \max(4,\lceil m\,pq \rceil) .
\end{equation*}
Within each panel, the horizontal axis is the non-separability index $\rho$ of \eqref{eq:rho} and the vertical axis is empirical power. Each grid is produced once for Gaussian data and once for heavy-tailed data: i.i.d. observations are generated from either a centered matrix-normal distribution \citep{gupta2018} or a centered matrix-$t$ distribution \citep{thompson2020classification} with covariance $\Sigma$.

The separable bins ($\rho=0$), together with the $w=0$ column of Section~\ref{sec:sim-reproduce}, provide the level check, visualized by the dashed black line at $\alpha = 0.05$: their rejection rates estimate the empirical size, so any distortion induced by the Gaussian calibration under the matrix-$t$ data is visible at the left edge of each panel.

By design, we get exact level control for the Gaussian setting shown in Figure~\ref{fig:simulations_setup1}. In all settings, the power of the tests increases with increasing non-separability $\rho$. 
The power also increases with increasing $m$ and thus $N$. For the non-Gaussian case shown in Figure~\ref{fig:simulations_setup1_t}, only the angular test maintains an approximately nominal size, and it does so increasingly well as the degrees of freedom grow. A mild residual size inflation remains for the heaviest tails and the smallest sample sizes, because the variability of the MMLE whitening under heavy-tailed data is undervalued by the Monte Carlo reference based on Gaussian data.

These results admit a coherent interpretation in the light of the theory developed above. Under Gaussian data (Figure~\ref{fig:simulations_setup1}), the empirical size sits on the nominal level in every panel, as guaranteed by Corollary~\ref{cor:exact_level}. Crucially, the power curves of the angular test are virtually indistinguishable from those of the elliptical test in every panel, confirming that the radial normalization costs essentially nothing, in line with the intuition that the scale information was already removed by the whitening step. Under matrix-$t$ data with Gaussian calibration (Figure~\ref{fig:simulations_setup1_t}), the elliptical test shows severe size inflation, growing as the degrees of freedom decrease. This is precisely the mechanism quantified in Proposition~\ref{prop:oracle-elliptical-radial-sensitivity} in Appendix~\ref{sec:oracle-angular-robustness}: the heavy-tailed radial component displaces the statistic by a distribution-specific amount relative to a Monte Carlo reference cloud whose spread is an order of magnitude smaller. The angular test removes the radial component and hence the displacement, to some extent, and accordingly maintains an approximately nominal size for misspecified distributions. The mild residue for the heaviest tails and smallest sample sizes is attributable to the remaining channel through which the data distribution enters, namely the variability of the estimated MMLE whitening, which is larger under heavy tails than under the Gaussian reference (cf.~Appendix~\ref{sec:oracle-angular-robustness}).

\begin{figure}
    \centering
    \includegraphics[width=1\linewidth]{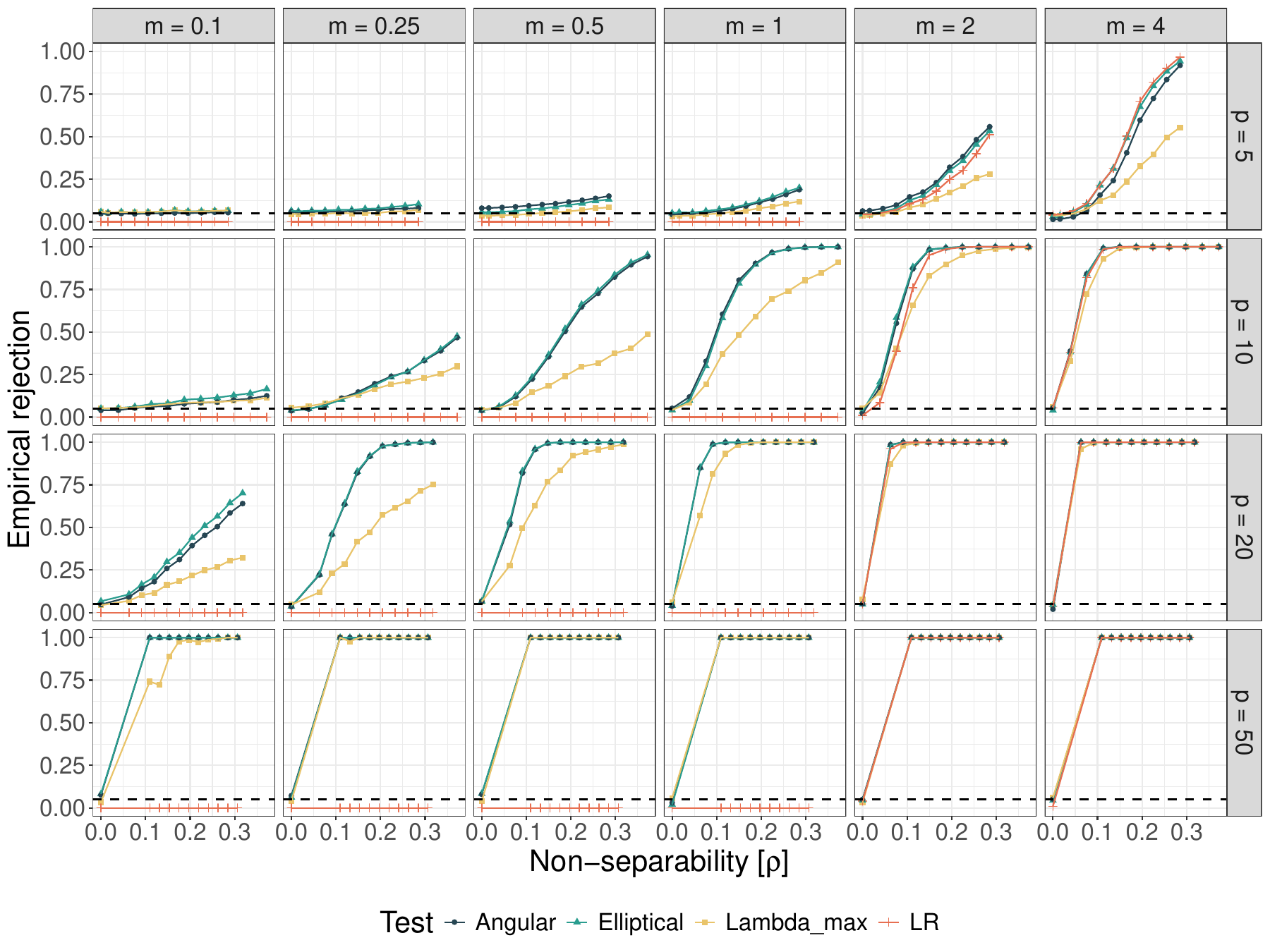}
    \caption{Empirical power vs.~the non-separability index for the Gaussian distribution. The sample size is fixed for each panel at $N = m\;pq$.}
    \label{fig:simulations_setup1}
\end{figure}

\begin{figure}
    \centering
    \includegraphics[width=1\linewidth]{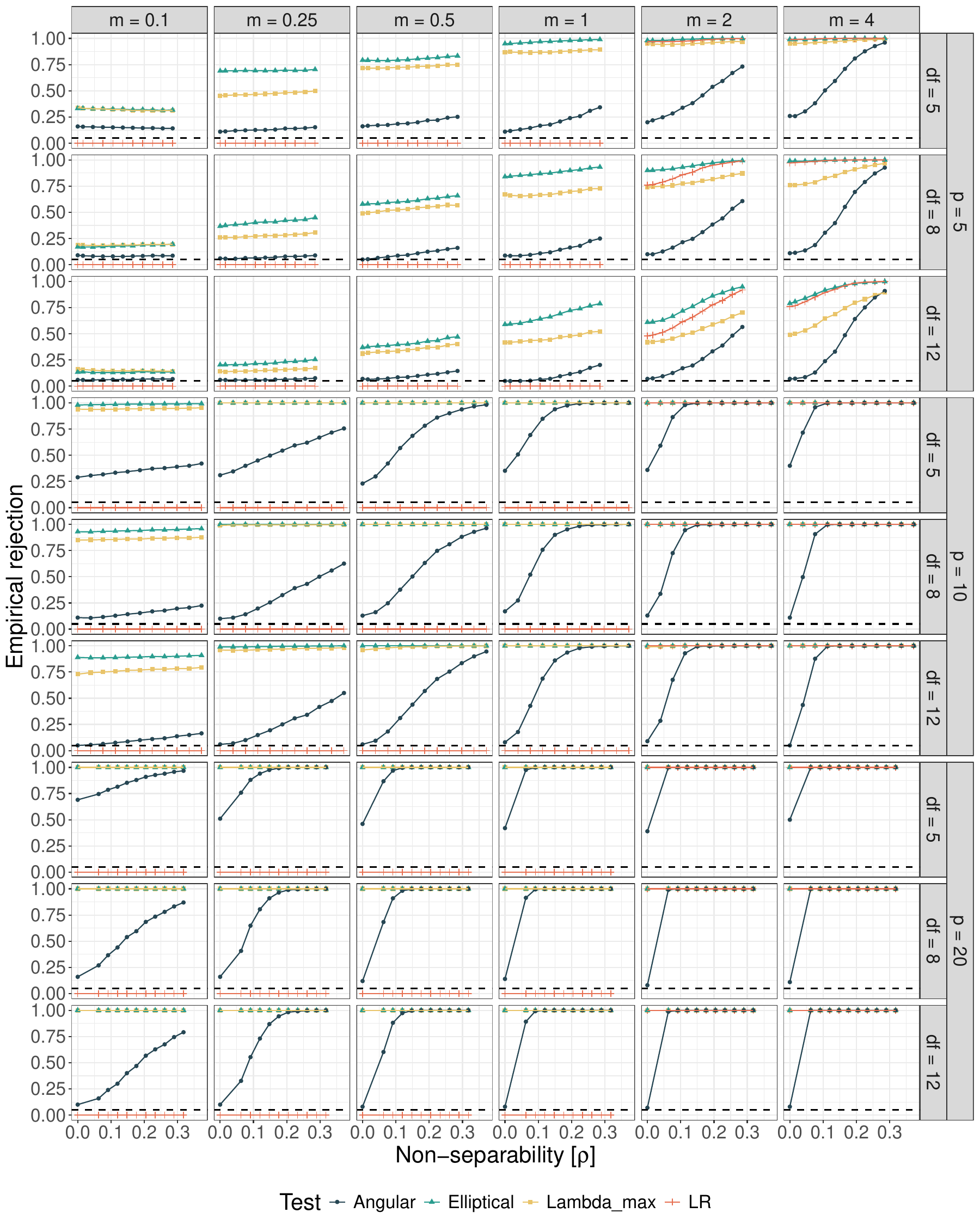}
    \caption{Empirical power vs.~the non-separability index for the $t$-distribution with degrees of freedom $5$, $8$, or $12$. The sample size is fixed for each panel at $N = m\;pq$.}
    \label{fig:simulations_setup1_t}
\end{figure}

\newpage
\subsection{Application: Acoustic Phonetic Data}
\label{sec:phonetics}

As a real-data illustration, we revisit the acoustic phonetic dataset of \citet{pigoli2014distances}.
The data consist of $219$ speech recordings from $23$ speakers of five Romance languages, French, Italian, Portuguese, American Spanish, and Iberian Spanish, each pronouncing the words for the numbers one to ten. The design is unbalanced, as not all digits are available for all speakers, so the number of recordings differs across languages, ranging from $25$ to $60$ (reported as the sample sizes $N_l$ in Table~\ref{tab:phonetics}). Recordings are sampled at $16{,}000$ samples per second. Relevant phonetic features of a recorded sound are most naturally explored in a time--frequency representation, and we consider two such representations of each recording.

\textbf{Log-spectrograms.}
We use the same log-spectrograms as \citet{aston2017tests}. They mapped each recording to a two-dimensional log-spectrogram via a short-time Fourier transform with a Gaussian window of width $10$\,ms. The raw log-spectrograms were non-linearly time-aligned after smoothing to remove phase distortion induced by differences in speaking rate. The resulting log-spectrograms were sampled on an equispaced grid of $81$ frequencies $\times$ $100$ time bins.

\textbf{Mel-frequency cepstral coefficients (MFCCs).}
As a second representation, we summarize each recording by its first $12$ MFCCs, resampled on a common feature grid of size $99$.  The MFCCs were computed using the \texttt{R} package \texttt{tuneR} \citep{tuneR} with the default settings and the frame/window length of 10 ms, which implements the usual sequence of preprocessing steps \citep[see][for details]{sueur2018sound}: short-time spectral analysis of the waveforms, transformation to 40 mel-filterbank energies, logarithmic compression, and a discrete cosine transform applied across the mel-frequency bands. This yields, for each recording, a $12\times 99$ matrix of cepstral features.

For both representations, a datum is a matrix-valued 
observation indexed by a \emph{frequency-type} direction and a binned \emph{temporal} direction. Estimating the full covariance is delicate: in either representation, the number of recordings available within a single language is far smaller than the number of free parameters of an unstructured covariance. A \emph{separable} covariance, which factorizes into a frequency-direction component and a temporal component as in $\Sigma=\Sigma_{1}\ot\Sigma_{2}$, involves dramatically fewer parameters and is therefore an attractive modeling assumption~\citep{pigoli2018statistical}. The purpose of this application is to assess that assumption with the angular test proposed in Section~\ref{sec:angular}. The test is carried out separately for each language $l=1,\dots,5$, and for the two matrix-variate representation of the data.

We apply the angular test to each language and each representation separately, calibrating the null distribution with $999$ Monte Carlo samples. A central motivation for our test is its behavior away from Gaussianity. Acoustic data are well known to be non-Gaussian, and a separability test whose level is controlled only under Gaussianity risks rejection merely because the data-generating distribution departs from the normal model rather than because the covariance is non-separable. This is precisely why \citet{aston2017tests} resorted to a nonparametric bootstrap calibration. Our angular test addresses this concern by design: while exact level control can only be guaranteed under Gaussianity, its nominal level shows markedly better distributional robustness than the alternatives considered in Section~\ref{sec:sim}. Rejections reported below are therefore unlikely to be artifacts of distributional misspecification.

The results are unambiguous. For every language and under both representations, the test statistic computed from the observed data exceeded all $999$ Monte Carlo samples, so that none of the simulated statistics was as extreme as the observed one. This yields the smallest $p$-value attainable with this calibration, namely $\widehat{p} = 1/(999+1) = 0.001$,
for all ten tests (Table~\ref{tab:phonetics}).

\begin{table}[t]
	\centering
	\caption{Monte Carlo $p$-values for the test of separability of the within-language covariance of the log-spectrograms and MFCCs of the five Romance languages, based on $999$ Monte Carlo samples. In every case, the observed test statistic exceeded all Monte Carlo samples. The sample sizes $N_l$ (number of recordings per language) are also reported.}

    \smallskip
	\label{tab:phonetics}
	\begin{tabular}{lccccc}
		\toprule
		Representation & French & Italian & Portuguese & Am.\ Spanish & Ib.\ Spanish \\
		\midrule
		Sample size $N_l$ & $60$ & $50$ & $25$ & $46$ & $38$ \\
		\midrule
        Log-spectrograms & $ 0.001$ & $ 0.001$ & $ 0.001$ & $ 0.001$ & $ 0.001$ \\
		MFCCs            & $ 0.001$ & $ 0.001$ & $ 0.001$ & $ 0.001$ & $ 0.001$ \\
		\bottomrule
	\end{tabular}
\end{table}

These findings strengthen the conclusion of \citet{aston2017tests}, who rejected separability for the log-spectrograms of the same languages using an empirical-bootstrap test, and extend it to the MFCC representation, which, to our knowledge, has not been previously considered for this data set. A further advantage over \citet{aston2017tests} is that the angular test requires no tuning of the number of projection directions: their bootstrap test is applied to a nested family of eigen-projection index sets and can reach different conclusions depending on that choice. For instance, separability is not rejected for French when only the leading directions are retained, and is rejected only once sufficiently many directions are included. The angular test yields a single, tuning-free decision that rejects in every language and representation. The fact that separability is rejected under two quite different time-frequency representations suggests that the non-separability is a genuine feature of the within-language acoustic variation rather than an artifact of any particular representation.

\section{Discussion}

We study tests for covariance separability for high-dimensional matrix-variate data, when the row and column dimensions are proportional to the sample size, as in $pq\asymp N$. Within the elliptical framework, separability can be reformulated as a sphericity problem after whitening by a separable covariance estimate. This converts the original testing problem, which would otherwise require reasoning about an unstructured covariance tensor, into a test for whether the whitened observations have covariance proportional to the identity. The first test we discuss, i.e. the \emph{elliptical test}, uses the sample variance of the whitened eigenvalues as the test statistic. The resulting procedure is computationally aligned with the purpose of separability itself: it avoids forming or storing the full empirical covariance tensor and can be implemented through inner products of the whitened observations.

\newpage
Interestingly, after matrix whitening, several classical sphericity statistics coincide with the \emph{elliptical} testing strategy. In particular, the matrix-whitened versions of John’s, Nagao’s, and Ledoit–Wolf’s statistics reduce to the same test statistic up to constants handled automatically by Monte Carlo calibration. This gives a single test statistic while preserving the intuition behind all these classical procedures. Under the null hypothesis, Monte Carlo calibration yields finite-sample level control when the calibration distribution is correctly specified, and the separable covariance estimator is uniquely defined and matrix affine equivariant. The matrix normal maximum likelihood estimator~\eqref{eq:mmle} satisfies these requirements under mild conditions, making it a natural choice for the whitening step. The test is shown to be consistent in the high-dimensional regime w.r.t.~dense alternative.

The empirical results highlight both the strength and the limitations of the elliptical test.
Specifically, Gaussian calibration of the elliptical test can be unreliable under heavy-tailed matrix-variate distributions. In that case, rejection may reflect merely tail misspecification rather than genuine non-separability.
This motivates the key practical refinement: an \emph{angular} version of the test, which is the procedure we recommend for practical use. By projecting the whitened observations onto the unit sphere, the angular test mitigates the distribution-specific effect of the radial component, thereby making the test more robust to miscalibration at the MC level. Overall, the angular test uses a similar strategy to the elliptical one, just using the standard deviation of the eigenvalues of the spatial-sign covariance of the whitened data as the test statistic.

The simulations show that angular projection indeed substantially improves robustness under heavy-tailed elliptical distributions with misspecified MC calibration, while retaining nearly the same power as the non-angular test under Gaussian data. In particular, for matrix-$t$ data, unlike the elliptical test, the angular test respects the level in almost all settings, with violations in the very heavy-tailed, small-sample regimes. This effect is consistent with the fact that the whitening factors are still estimated from the data, so, for small sample sizes, their variability may remain distribution-dependent. 

Several open questions remain.
Firstly, the finite-sample level guarantees are not tied to the asymptotic regime, but the consistency of the tests is. Extending the power analysis to ultra-high-dimensional settings would be a natural future step. Secondly, the consistency results rely on convergence assumptions for the separable maximum likelihood estimator under alternatives. These assumptions are natural and supported empirically, but the corresponding theoretical guarantees for the pseudo-true separable parameter remain open. Finally, while the angular test reduces sensitivity to radial misspecification within the elliptical family, it does not eliminate all distributional dependence, because the whitening factors are estimated.%

\newpage

\bibliographystyle{plainnat}
\bibliography{refs}

\appendix
\renewcommand{\thesection}{\Alph{section}}
\makeatletter
\def\@seccntformat#1{\ifcsname AppPrefix@#1\endcsname\csname AppPrefix@#1\endcsname\else\csname the#1\endcsname\quad\fi}
\newcommand\AppPrefix@section{Appendix \thesection\quad}
\makeatother

\section{Additional Simulations} \label{appendix:sim}

\subsection{Convergence rates of MMLE}
\label{sec:appendix_sim}
Figure~\ref{fig:mmle_convergence} shows the Frobenius norm rate 
$$\frac{1}{\sqrt{pq}}\|\widehat{\Sigma}_1 \ot\widehat{\Sigma}_2 - \Sigma_1 \ot \Sigma_2 \|_F$$ 
and the operator norm rate 
$$\|\widehat\Sigma_{1}\ot \widehat\Sigma_{2} -\Sigma_{1}\ot \Sigma_{2}\|_{2}$$
for Gaussian data based on the $1000$ non-separable covariances of Section~\ref{sec:empirical}, where $(\Sigma_1,\Sigma_2)$ is the pseudo-true parameter defined in~\eqref{eq:pseudo-parameter}. The separable errors are based on data generated from the best separable rank-one approximation of the non-separable covariances.

These results numerically verify that the rates for the pseudo-parameters are comparable to the ones provided by \citet{franks2026near} under separability, which are actually stronger than what assumptions (A4) and (B4) of Theorems~\ref{thm:consistency} and~\ref{thm:angular_consistency}, respectively, impose.

\begin{figure}[H]
	\centering
	\includegraphics[width=1\linewidth]{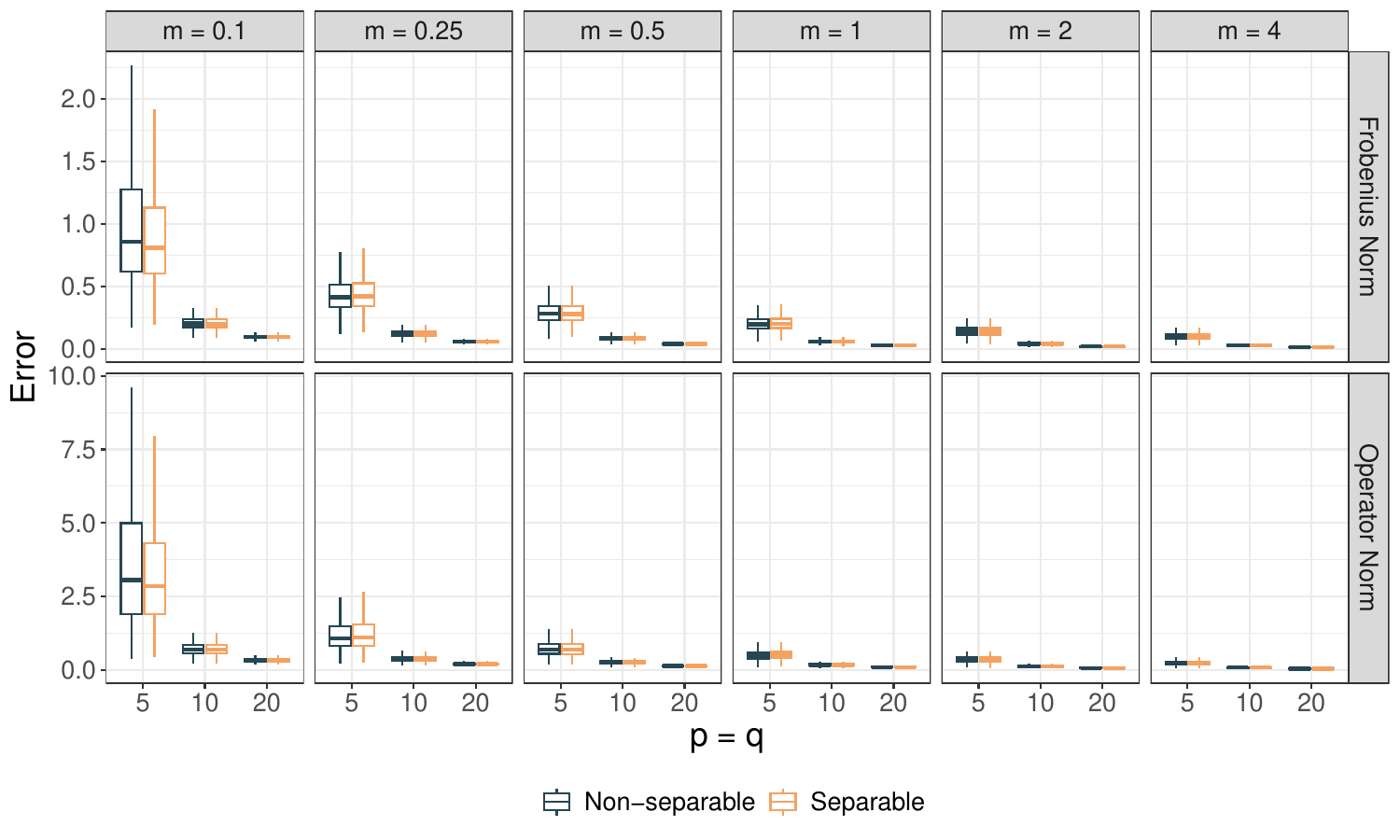}
	\caption{Convergence rates of MMLE to the pseudo-parameters. The box plots do not include points outside of the whiskers.}
	\label{fig:mmle_convergence}
\end{figure}

\subsection{Comparison with the Functional Data Test}
\label{sec:sim_aston}

In Figure~\ref{fig:simulations_setup1_Aston}, we analyze the performance of the projection-based empirical bootstrap test (as per the suggestion of the authors) of \cite{aston2017tests} with $l \in \{1,2,3\}$ projections in both directions, using the same setup as in Section~\ref{sec:sim-family} for $p = q = 10$. While the empirical bootstrap test (EBT) respects the level, it lacks the power of the other tests. This is likely because it is based on covariance estimators \citep[i.e. the partially traced ones of][rather than the MMLEs we opt for]{aston2017tests} that are less efficient in the regime considered here. Another reason is presumably that the EBT does not use the full spectral information available, since it truncates the ranks in both the row and column space to $l$ PCA-defined projections.

\begin{figure}[!h]
    \centering
    \includegraphics[width=1\linewidth]{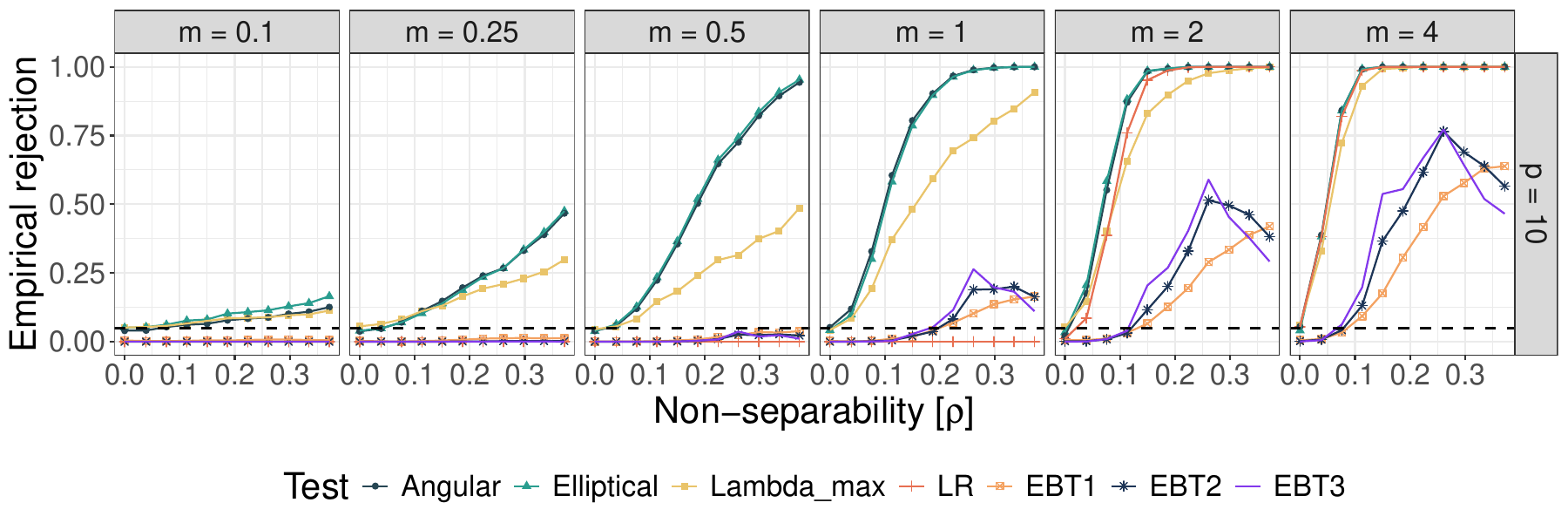}
    \caption{Empirical power vs.~the non-separability index for the Gaussian distribution. The sample size is fixed for each panel at $N = m\;pq$.}
    \label{fig:simulations_setup1_Aston}
\end{figure}

\section{Proofs}

\subsection{Proof of Lemma~\ref{lem:statistics_equal}}

We start by showing that $\tr(\widehat C)=pq$. Before we do that, let us treat the population level first.

\begin{lemma}\label{lem:trace}
    For the whitened population-level covariance
    \[
    C = (\Sigma_1 \ot \Sigma_2)^{-1/2} \Sigma (\Sigma_1 \ot \Sigma_2)^{-1/2},
    \]
    it holds $\tr(C)=pq$.
\end{lemma}
\begin{proof}
Let us denote $\widetilde{Y}_n := \Sigma_1^{-1/2} X_n \Sigma_2^{-1/2}$ the oracle whitened data, where $\Sigma_1,\Sigma_2$ satisfy the population KKT conditions for MMLE, i.e., the fixed-point equations~\eqref{eq:pseudo-parameter}
\[
\frac{1}{q} \mathbb{E} X_1 \Sigma_2^{-1} X_1^\top = \Sigma_1, \qquad   \frac{1}{p} \mathbb{E} X_1^\top \Sigma_1^{-1} X_1 = \Sigma_2.
\]
Note that $C = \cov(\widetilde{Y}_1)$, and the population KKT conditions for MMLE based on $\widetilde{Y}_1$ are
\[
\frac{1}{q} \mathbb{E} \widetilde{Y}_1 \widetilde{Y}_1^\top = I_p,\qquad \frac{1}{p} \mathbb{E} \widetilde{Y}_1^\top \widetilde{Y}_1 = I_q,
\]
since the MMLE based on $\widetilde{Y}_1$ has to yield identity matrices due to the invariance property of the MLE estimators. Using the first equation, we get that $\mathrm{tr}(\mathbb{E} \widetilde{Y}_1 \widetilde{Y}_1^\top)=pq$. At the same time
\[
\mathrm{tr}(\mathbb{E} \widetilde{Y}_1 \widetilde{Y}_1^\top)=\mathbb{E} \mathrm{tr}(\widetilde{Y}_1 \widetilde{Y}_1^\top)= \mathbb{E} \| \widetilde{Y}_1 \|_F^2 = \mathrm{tr}(C),
\]
which proves the claim.
\end{proof}

The proof of $\tr(\widehat C)=pq$ is almost identical, only taking place on the sample level. The KKT conditions for the MMLE can be written using the whitened sample as
\[
\frac{1}{q}\left[ \frac{1}{N} \sum Y_n Y_n^\top \right] = I_p \qquad \& \qquad \frac{1}{p} \left[ \frac{1}{N} \sum Y_n^\top Y_n \right] = I_q,
\]
yielding $\mathrm{tr}( \frac{1}{N} \sum Y_n Y_n^\top )=pq$, but also
\[
\mathrm{tr}\left( \frac{1}{N} \sum Y_n Y_n^\top \right) = \frac{1}{N} \sum \mathrm{tr}(Y_n Y_n^\top) = \frac{1}{N} \sum \|Y_n\|_F^2 = \frac{1}{N} \sum \mathrm{tr}(Y_n \otimes Y_n) = \mathrm{tr}(\widehat{C}).
\]

Nagao's and John's statistics are thus exactly equal, while the Ledoit-Wolf statistic differs from it by the deterministic additive constant. The inner product representation now follows straightforwardly. Starting e.g.~from the Nagao's statistic:
\[
T_N^{(\mathrm{Na})}
=
\frac{1}{pq}\|\widehat C-I_{p,q}\|_F^2
=
\frac{1}{pq}
\left\{
\operatorname{tr}(\widehat C^2)
-2\operatorname{tr}(\widehat C)
+\operatorname{tr}(I_{p,q})
\right\}.
\]
Using \(\operatorname{tr}(\widehat C)=pq\) and
\(\operatorname{tr}(I_{p,q})=pq\), this simplifies to
\[
T_N^{(\mathrm{Na})}
=
\frac{1}{pq}\operatorname{tr}(\widehat C^2)-1.
\]
At the same time, we have
\[
\operatorname{tr}(\widehat C^2)
=
\left\langle \widehat C,\widehat C \right\rangle_F
=
\frac{1}{N^2}
\sum_{n=1}^N
\sum_{n'=1}^N
\left\langle
Y_n\otimes Y_n,\,
Y_{n'}\otimes Y_{n'}
\right\rangle_F .
\]
Finally, for rank-one tensors, we have
\[
\left\langle
Y_n\otimes Y_n,\,
Y_{n'}\otimes Y_{n'}
\right\rangle_F
=
\langle Y_n,Y_{n'}\rangle_F^2,
\]
which completes the proof of Lemma~\ref{lem:statistics_equal}.

\subsection{Proof of Theorem~\ref{thm:law}}

Under the null, $X_n = (\Sigma_1\ot\Sigma_2)^{1/2} Z_n$. Let $\widetilde\Sigma_1\ot\widetilde\Sigma_2$ denotes the separable covariance estimator based on $\{Z_n\}$, while $\widehat\Sigma_1\ot\widehat\Sigma_2$ still denotes the one based on $\{X_n\}$. By equivariance,
\begin{equation}\label{eq:equivar}
    \widehat\Sigma_1\ot\widehat\Sigma_2 = (\Sigma_1\ot\Sigma_2)^{1/2} ( \widetilde\Sigma_1\ot\widetilde\Sigma_2)  (\Sigma_1\ot\Sigma_2)^{1/2} 
\end{equation}

Write $A:=\mat(\widetilde\Sigma_1\ot\widetilde\Sigma_2) \in \R^{pq \times pq}$ and $B := \mat((\Sigma_1\ot\Sigma_2)^{1/2})\in \R^{pq \times pq}$ for short, and define
\[
Q := (BAB)^{-1/2} B A^{1/2}. 
\]
Then
\[
Q Q^\top=(BAB)^{-1/2}BAB (BAB)^{-1/2}=I,
\]
so $Q$ is orthogonal and
\[
(BAB)^{-1/2}B=Q A^{-1/2}.
\]
Therefore, whitening of $x_n = \vec(X_n)$ using the RHS of equation~\eqref{eq:equivar} gives
\[
(B A B)^{-1/2}(x_n) = (B A B)^{-1/2}(B (z_n)) = (Q A^{-1/2})(z_n) = Q(A^{-1/2}(z_n)),
\]
which is just an orthogonal transformation of the whitening of $Z_n$.

We used the matricizations above for simplicity, to be able to use matrix transposes instead of adjoint operators, which would need to be defined. In the rest of the proof, working with tensors is simpler, so let us re-define our short-hands as $A:=\widetilde\Sigma_1\ot\widetilde\Sigma_2$ and $B := (\Sigma_1\ot\Sigma_2)^{1/2}$. Now, we have
\[
    BAB = \big(\Sigma_1^{1/2}\widetilde\Sigma_1\Sigma_1^{1/2}\big) \ot \big(\Sigma_2^{1/2}\widetilde\Sigma_2\Sigma_2^{1/2}\big).
\]
    Using the square-root identity of Kronecker product $(M_1\ot M_2)^{1/2} = (M_1^{1/2}\ot M_2^{1/2})$,
\[
    (BAB)^{-1/2} = \big(\Sigma_1^{1/2}\widetilde\Sigma_1\Sigma_1^{1/2}\big)^{-1/2} \ot \big(\Sigma_2^{1/2}\widetilde\Sigma_2\Sigma_2^{1/2}\big)^{-1/2}, \qquad A^{1/2} = \widetilde\Sigma_1^{1/2} \ot \widetilde\Sigma_2^{1/2}.
\]
    Substituting into the definition of $Q$ and applying the mixed-product property twice more,
\[
    Q = (BAB)^{-1/2}BA^{1/2} = Q_1 \ot Q_2, \quad\text{where}\quad Q_i := \big(\Sigma_i^{1/2}\widetilde\Sigma_i\Sigma_i^{1/2}\big)^{-1/2}\Sigma_i^{1/2}\widetilde\Sigma_i^{1/2}, \quad i=1,2.
\]
Therefore, $Q_1Q_1^\top=I_p$ and $Q_2Q_2^\top=I_q$ are orthogonal matrices.

Hence, the whitening $\widetilde X_n$ of $X_n$ is just an orthogonal transformation of the whitening $\widetilde Z_n$ of $Z_n$, i.e.~$\widetilde X_n=Q_1\widetilde Z_nQ_2^\top$, $n=1,\dots,N$. Thus, for any orthogonally invariant test statistic $T_N$
is $T_N(\{\tilde X_n\})\sim T_N(\{\tilde Z_n^\star\})$, where $\tilde X_n$ and $\tilde Z_n$ are the whitenings of $X_n$ and $Z_n$, respectively, and $\tilde Z_n^\star\sim \tilde Z_n$.

\paragraph{Proof of Corollary~\ref{cor:cor_of_thm2}
} Following the notation of the proof of Theorem \ref{thm:law}, the empirical covariance of the whitened $\{X_n\}$ is just an orthogonal transformation of the empirical covariance of the whitened $\{Z_n\}$, specifically
\begin{equation}\label{eq:invariance}
\frac{1}{N} \sum (BAB)^{-1/2}(x_n) \otimes (BAB)^{-1/2}x_n
= Q \left(\frac{1}{N} \sum A^{-1/2}(z_n) \otimes A^{-1/2}z_n\right) Q^\top.
\end{equation}
Hence the two whitened covariances have the same eigenvalues, and so we have $T_N(\{X_n\}) \sim T_N(\{Z_n^\star\})$ for the spectral statistic, with $\{Z_n^\star\}$ generated from the distribution of $\{Z_n\}$. 

\subsection{Proof of Proposition~\ref{prop:MMLE}}

Firstly, under the sample size condition, the MMLE exists and is uniquely defined (up to the usual scale indeterminacy, which however vanishes in the product) almost surely by \citet{drton2021}. Secondly, the matrix equivariance of of the MMLE follows easily from the outer product structure and invariance of maximum likelihood estimators; for the proof see e.g.~\citet[choosing $h=n$ in Lemma 3.0.1]{MMCD} and also \citet{eaton1989group}.

On the other hand, the partial tracing or the best separable approximation estimators are not matrix equivariant, since they are tied to the ambient Euclidean coordinates. These estimators commute with orthogonal changes of bases, but not with general invertible linear transformations.

We first show this for partial tracing for which
\[
\widehat\Sigma_1 = \frac{1}{qN}\sum_{n=1}^N X_nX_n^\top \qquad\& \qquad \widehat\Sigma_2 = \frac{1}{pN}\sum_{n=1}^N X_n^\top X_n .
\]
Under a general separable linear transformation $\widetilde X_n := (A \ot B)X_n$, we obtain
\[
\widetilde\Sigma_1 = \frac{1}{qN}\sum_{n=1}^N AX_n B^\top BX_n^\top A^\top \qquad\& \qquad \widetilde\Sigma_2 = \frac{1}{pN}\sum_{n=1}^N B^\top X_n^\top A A^\top X_n B,
\]
which are generally not equal (even up to scale) to
\[
A\left(\frac{1}{qN}\sum_{n=1}^N X_n X_n^\top \right) A^\top \qquad\& \qquad B^\top\left(\frac{1}{pN}\sum_{n=1}^N  X_n^\top X_n \right) B.
\]

For the best separable approximation estimator
\[
(\widehat\Sigma_1,\widehat\Sigma_2) = \argmin_{\Sigma_1,\Sigma_2} \| \widehat{\Sigma} - \Sigma_1 \ot \Sigma_2 \|_F^2,
\]
the situation is similar. Frobenius norm is invariant under orthogonal transformations only. Specifically, forming the empirical covariance with a general separable linear transformation $\widetilde X_n := (A \ot B)X_n$, we see immediately that
\[
\| (A \ot B) \widehat{\Sigma} (A\ot B)^\top - \Sigma_1 \ot \Sigma_2 \|_F^2 \neq \| \widehat{\Sigma} - (A\ot B)(\Sigma_1 \ot \Sigma_2) (A \ot B)^\top \|_F^2
\]
in general, where $\widehat{\Sigma}$ denotes the empirical covariance of $\{X_n\}$. The first term in left hand side of the previous inequality is precisely the empirical covariance of $\{\widetilde X_n\}$, since
\[
(A\ot B) X \otimes (A\ot B) X = [(A\ot B) \ot (A\ot B)](X \otimes X) = (A\ot B) (X \otimes X) (A\ot B)^\top,
\]
which follows immediately from the rank one element definition of $\ot$ in Section~\ref{sec:background}.

\subsection{Proof of Theorem~\ref{thm:consistency}}

Let us first recap the setup, and introduce some new notation that will be useful for the proof. Under the assumptions, $X_1,\ldots,X_N \stackrel{i.i.d}{\sim}(0,\Sigma)$, where the distribution is not necessarily Gaussian, and $Z_1,\ldots,Z_N \stackrel{i.i.d}{\sim}(0,I_{p,q})$. We work with separably-whitened data $Y_n = (\widehat\Sigma_1 \ot \widehat\Sigma_2)^{-1/2} X_n$ with the population-level covariance
\[
C  = \cov(\widetilde{Y}_1) := (\Sigma_1 \ot \Sigma_2)^{-1/2} \Sigma (\Sigma_1 \ot \Sigma_2)^{-1/2}.
\]
For the proof, we will also need the oracle version $\widetilde Y_n = (\Sigma_1 \ot \Sigma_2)^{-1/2} X_n$. We will denote by $\widehat{C}$ the empirical covariance of $\{Y_n\}$, and $\widetilde{C}$ the empirical covariance of $\{\widetilde Y_n\}$, i.e.
\[
\begin{split}
   \widehat C &= (\widehat\Sigma_1 \ot \widehat\Sigma_2)^{-1/2}\left[ \frac{1}{N}\sum_{n=1}^N X_n \otimes X_n\right] (\widehat\Sigma_1 \ot \widehat\Sigma_2)^{-1/2},\\
   \widetilde C &= (\Sigma_1 \ot \Sigma_2)^{-1/2} \left[ \frac{1}{N}\sum_{n=1}^N X_n \otimes X_n\right] (\Sigma_1 \ot \Sigma_2)^{-1/2}.
\end{split}
\]
Note that the previous expressions for the covariance only utilize linearity of the operators in conjunction with the following tensor product calculations:
\[
\begin{split}
A X B^\top \otimes CXD^\top &= (A \ot B) X \otimes (C \ot D) X = [(A\ot B) \ot (C \ot D)](X \otimes X)\\
& = (A \ot B) (X \otimes X) (C \ot D)^\top.
\end{split}
\]

Recall that due to Lemma~\ref{lem:statistics_equal} it is
\[
T_N = \frac{1}{pq}\|\widehat{C} - I_{p,q}\|_F^2 = \frac{1}{pq} \tr(\widehat{C}^2) -1
\]
in the case of John and Nagao, while for Ledoit-Wolf one removes $pq/N$ from the statistic.

Now, Assumption (A5), together with (A2), guarantees dense departure from sphericity on the whitened level. Specifically,
\[
\begin{split}
\widetilde{\Delta}_N :&= \frac{1}{\sqrt{pq}} \| C - I_{p,q} \|_F = \frac{1}{\sqrt{pq}} \left\| (\Sigma_1 \ot \Sigma_2)^{-1/2} (\Sigma - \Sigma_1 \ot \Sigma_2) (\Sigma_1 \ot \Sigma_2)^{-1/2} \right\|_F\\
&\geq \frac{\frac{1}{\sqrt{pq}}\|\Sigma - \Sigma_1 \ot \Sigma_2 \|_F}{\lambda_{\max}(\Sigma_1 \ot \Sigma_2)} \geq \frac{\Delta_0}{\lambda_{\max}(\Sigma_1) \lambda_{\max}(\Sigma_2)} =: \widetilde{\Delta}_0 > 0.
\end{split}
\]

If we show that $T_N - \widetilde{\Delta}_N^2 \pto0$, consistency is proven. The zero limit is pertinent to the Ledoit-Wolf statistic, while the limit will be $\gamma_1\gamma_2$ for John's and Nagao's statistic. However, constants do not matter in the Monte Carlo scheme, rather the existence of separation in important. The Monte Carlo test calibrates for the variability of $T_N$ calculated under the null around either zero or $\gamma_1\gamma_2$ (depending on the statistic used) with the fluctuation vanishing with increasing $N$, while $T_N$ calculated on the data under the alternative is separated from either zero or $\gamma_1\gamma_2$ by at least $\widetilde{\Delta}_0^2$, yielding consistency.

We start by introducing the oracle:
\[
T_N - \widetilde{\Delta}_N^2 = \frac{1}{pq} \tr(\widehat{C}^2-C^2)  = \frac{1}{pq}\tr(\widehat{C}^2-\widetilde{C}^2) + \frac{1}{pq}\tr(\widetilde{C}^2-C^2) =: (I) + (II).
\]
Let us deal with the term $(I)$ first:
\[
(I) \leq \frac{1}{\sqrt{pq}}\|\widehat{C} - \widetilde{C}\|_F \frac{1}{\sqrt{pq}}\left[ \| \widehat{C}\|_F + \|\widetilde{C}\|_F \right].
\]
Both terms in the brackets are $\| \Sigma \|_F+\mathcal{O}_p(1)$, i.e.~they are $\mathcal{O}_p(1)$ after the multiplication by $(pq)^{-1/2}$. As for $\|\widehat{C} - \widetilde{C}\|_F$, notice we can write
\[
\widehat{C} - \widetilde{C} = A S A - BSB
\]
where $A:=(\widehat\Sigma_1 \ot \widehat\Sigma_2)^{-1/2}$ and $B := (\Sigma_1 \ot \Sigma_2)^{-1/2}$. Re-parametrizing via $A=B+D$, we can write
\[
\widehat{C} - \widetilde{C} = DSB+BSD+DSD,
\]
and thus
\[
\begin{split}
    \|\widehat{C} - \widetilde{C}\|_F \leq &\| (\widehat\Sigma_1 \ot \widehat\Sigma_2)^{-1/2} - (\Sigma_1 \ot \Sigma_2)^{-1/2} \|_F \| \widehat{\Sigma}_N \|_{\text{op}} \| (\Sigma_1 \ot \Sigma_2)^{-1/2} \|_{\text{op}} +\\
    &\| (\Sigma_1 \ot \Sigma_2)^{-1/2} \|_{\text{op}} \| \widehat{\Sigma}_N \|_{\text{op}} \| (\widehat\Sigma_1 \ot \widehat\Sigma_2)^{-1/2} - (\Sigma_1 \ot \Sigma_2)^{-1/2} \|_F +  \\
    &\| (\widehat\Sigma_1 \ot \widehat\Sigma_2)^{-1/2} - (\Sigma_1 \ot \Sigma_2)^{-1/2} \|_F^2 \| \widehat{\Sigma}_N \|_{\text{op}}.
\end{split}
\]
Since $\| \widehat{\Sigma}_N \|_{\text{op}}$ and $\| (\Sigma_1 \ot \Sigma_2)^{-1/2} \|_{\text{op}}$ are both $\mathcal{O}_p(1)$, while the Frobenius norm term is $o_p(1)$ after multiplication with the remaining $(pq)^{-1/2}$ as per Assumption (A4), we have overall that $(I) = o_p(1)$.

For the term $(II)$, we will show it is equal to $\frac{pq}{N} + o_p(1)$ via mean-variance calculations. At this point, it is beneficial to vectorize: let $x_n = \vec(X_n)$ and $z_n = \vec(Z_n)$ both in $R^{pq}$, i.e.~$x_n = \mat(\Sigma)^{1/2}\, z_n$, where $\mat(\Sigma) \in \R^{pq \times pq}$. Denote
\[
\Gamma := \mat(\Sigma)^{1/2} (\Sigma_2 \otimes_K \Sigma_1)^{-1}\mat(\Sigma)^{1/2},
\]
which is just a matrix representation of $\Gamma_{\mathrm{tensor}} := \Sigma^{1/2} (\Sigma_1 \ot \Sigma_2)^{-1}\Sigma^{1/2}$, so in particular traces of their powers are the same \citep{golub2013matrix}. And since it is straightforward to verify that $\tr(\Gamma_\mathrm{tensor}^2)=\tr(C^2)$, it is also $\tr(\Gamma^2)=\tr(C^2)$. Denoting $A_{n,m} := z_n^\top \Gamma z_m$, we can write
\[
\begin{split}
\mathrm{tr}(\widetilde C^2) &= \langle \widetilde C, \widetilde C \rangle = \frac{1}{N^2} \sum_{n=1}^N \sum_{m=1}^N \langle {\Sigma}_1^{-1/2} X_n {\Sigma}_2^{-1/2}, {\Sigma}_1^{-1/2} X_m {\Sigma}_2^{-1/2} \rangle \\
&= \frac{1}{N^2} \sum_{n=1}^N \sum_{m=1}^N \langle ({\Sigma}_1 \ot {\Sigma}_2)^{-1/2} X_n , ({\Sigma}_1 \ot {\Sigma}_2)^{-1/2} X_m \rangle\\
&= \frac{1}{N^2} \sum_{n=1}^N \sum_{m=1}^N \langle ({\Sigma}_2 \otimes_K {\Sigma}_1)^{-1/2} x_n , ({\Sigma}_2 \otimes_K {\Sigma}_1)^{-1/2} x_m \rangle \\
&= \frac{1}{N^2} \sum_{n=1}^N \sum_{m=1}^N \big(x_n^\top (\Sigma_2 \otimes_K \Sigma_1)^{-1} x_m\big)^2= \frac{1}{N^2} \sum_{n=1}^N \sum_{m=1}^N A_{n,m}^2.
\end{split}
\]

We begin with the mean calculations. For $m\neq n$, $z_n$ and $z_m$ are independent and thus
\[
\mathbb{E}A_{n,m}^2 = \mathbb{E} \mathbb{E} [z_n^\top \Gamma z_m z_m^\top \Gamma z_n| z_n] = \mathbb{E}z_n^\top \Gamma \mathbb{E}[z_m z_m^\top|z_n] \Gamma z_n,
\]
where the conditional expectation in the middle is just $\mathrm{Cov}(z_m)=I$. Thus the contribution of the off-diagonal entries is
\[
\mathbb{E}A_{n,m}^2 = \mathbb{E}z_n^\top \Gamma^2 z_n = \mathrm{tr}(\Gamma^2)=\mathrm{tr}(C^2).
\]
For the diagonal entries, we will use an explicit fourth moment bound, so let $\mu_4:=\E z_{1,1}^4$. In $z_n^\top \Gamma z_n=\sum_{i,j} \gamma_{i,j} z_{n,i}z_{n,j}$, we will drop the index $n$ for better readability. Then $\E A_{n,n}^2 = \sum_{i,j,k,l} \gamma_{i,j} \gamma_{k,l} \E z_i z_j z_k z_l$.

Now it is time to carry out standard pairing argument. The expectation $\E z_i z_j z_k z_l$ zeroes out unless either $i=j=k=l$ (resulting in the 4th moment $\mu_4$) or two pairs of the indices are the same (resulting in the variance squared, which is equal to one here).
\begin{itemize}
  \item When $i=j$ and $k=l$, the sum reduces to $\tr(\Gamma)^2$, which is again just $\tr{C}^2$.
  \item When $i=k$ and $j=l$ (or alternatively $i=l$ and $j=k$), the sum reduces to $\tr(\Gamma^2) = \tr(C^2)$, which is added twice due to the alternative option.
  \item When $i=j=k=l$, the sum reduced to $\mu_4 \sum_i \gamma_{i,i}^2 \leq \mu_4 \tr(\Gamma^2)$ which is again just $\mu_4 \tr(C^2)$.
\end{itemize}
We have to remember that the $i=j=k=l$ case is a sub-case of the other three cases, so in total we have
\[
\E A_{n,n}^2 = \tr(C)^2 + 2 \tr(C^2)+(\mu_4-3)\sum_i \gamma_{i,i}^2.
\]
Overall, we thus have
\[
\E\, \mathrm{tr}(\widetilde C^2) = \frac{N(N-1)}{N^2} \tr(C^2) + \frac{N}{N^2}\left[ \tr(C)^2 + 2 \tr(C^2) + (\mu_4 - 3) \sum_i \gamma_{i,i}^2 \right].
\]
Since $\tr(C)=pq$ by Lemma~\ref{lem:trace}, we have
\[
\frac{1}{pq}\E \mathrm{tr}(\widetilde C^2) = \frac{1}{pq}\tr(C^2)+\frac{pq}{N} + \mathcal{O}\left(\frac{1}{N}\right)
\]
and so the mean of $(II)$ has the desired limiting behavior.

Now we show that the variance vanishes by similar calculations. Firstly note that
\[
\mathrm{Var}(\mathrm{tr}(\widetilde C^2)) = \frac{1}{N^4} \sum_{m,n,m',n'} \cov(A_{n,m}^2,A_{n',m'}^2).
\]
For disjoint pairs of indices, the covariance is zero. However, we have to again separate three cases when $(m,n) \cap(m',n')\neq\emptyset$.

Case 1. For $n=m=n'=m'$, we have
\[
\cov(A_{n,m}^2,A_{n',m'}^2) \leq \E A_{n,m}^4 \leq \|\Gamma\|_{\text{op}}^4\E \|z_n\|_F^8 = \|C\|_\op^4 \mathcal{O}((pq)^4),
\]
but there is just $N$ such entries.

Case 2. For the perfectly off-diagonal entries, where $n=n'\neq m=m'$, by the Cauchy-Schwarz inequality
\[
|\cov(A_{n,m}^2,A_{n',m'}^2)| \leq \left( \E A_{n,m}^4 \E A_{n',m'}^4 \right)^{1/2} \leq \max_{n,m} \E A_{n,m}^4.
\]
Thus if we manage to show that $\max_{n,m} \E A_{n,m}^4 = \mathcal{O}(p^2q^2)$, we will have the vanishing variance $\mathrm{var}\big(\mathrm{tr}(\frac{1}{pq}\widetilde C^2)\big) = \mathcal{O}(\frac{1}{N})$ in this case.

To obtain $\max_{n,m} \E A_{n,m}^4 = \mathcal{O}(p^2q^2)$, we will use Rosenthal's bound \citep{rosenthal1970,merlevede2013rosenthal} for the fourth moment with an additional conditioning step:
\[
\max_{n,m} \E A_{n,m}^4 = \max_{n,m} \E (z_n^\top \Gamma z_m)^4 = \max_{n,m} \E \E[ (z_n^\top \Gamma z_m)^4|z_m].
\]
Given $z_m$, we can write $\E (z_n^\top \Gamma z_m)^4 = \E (\sum_k z_{n,k} u_k)^4$ where $u = \Gamma z_m$ is fixed. Rosenthal's bound be applied here to obtain
\[
\E[ (z_n^\top \Gamma z_m)^4|z_m] \leq K_4 \sum_{k} \E(z_{n,k} u_k)^4 + \left( \sum_k \E(z_{n,k} u_k)^2 \right)^2,
\]
where $K_4$ is a constant independent of everything (depending only on the power, which is fixed to four). But $u_k$ is fixed in this conditional calculation, so $\E(z_{n,k} u_k)^4=u_k^4 \mu_4$ $\E(z_{n,k} u_k)^2 = u_k^2$. Hence we have
\[
\E[ (z_n^\top \Gamma z_m)^4|z_m] \leq K_4 \mu_4 \sum_{k} u_k^4 + (\sum_k u_k^2)^2 \leq (K_4 \mu_4 + 1) \|u\|_2^4
\]
Let us plug in back for $u$ and integrate out the remaining randomness:
\[
\E (z_n^\top \Gamma z_m)^4 \leq (K_4 \mu_4 + 1) \|\Gamma z_m\|_2^4 \leq (K_4 \mu_4 + 1) \|\Gamma\|_2^4 \E \|z_m\|_2^4.
\]
Since $\Gamma$ and $C$ have the same non-zero eigenvalues, $\|\Gamma\|_2=\|C\|_2$ and hence
\[
\E \|z_m\|_2^4= \E (\sum_k z_{m,k}^2)^2=\sum_k\sum_l\E z_{m,k}^2 z_{m,l}^2=pq(pq-1)+\mu_4pq=\mathcal{O}(p^2q^2).
\]
In the perfectly off-diagonal case, there is at most
\[
\Big|\Big\{(m,n,m',n'): \{m,n\} \cap \{m',n'\} \neq \emptyset\Big\}\Big| \leq 4 N^3,
\]
entries.

Case 3. In the intermediate case, where exactly three of the indices are equal, assume w.l.o.g.~that $n\neq m=n'=m'$. Then we have
\[
\cov(A_{n,m}^2,A_{n',m'}^2) \leq \E (z_n^\top \Gamma z_m)^2(z_m^\top \Gamma z_m)^2,
\]
which, conditioning on $z_m$ is equal to
\[
\E z_m^\top \Gamma \underbrace{\E z_n z_n^\top}_{=I_{p,q}} \Gamma z_m (z_m^\top \Gamma z_m)^2.
\]
Since $\Gamma^2 \preceq \|\Gamma\|_\op \Gamma$, we have in turn
\[
\cov(A_{n,m}^2,A_{n',m'}^2) \leq \|\Gamma\|_\op \E (z_m^\top \Gamma z_m)^3 \leq \|\Gamma\|_\op^4 \E \|z_m\|_F^6 = \mathcal{O}(p^3q^3),
\]
and there is just $\mathcal{O}(N^2)$ such entries.

Overall, we have
\[
\mathrm{Var}(\mathrm{tr}(\widetilde C^2)) = \mathcal{O}\left(\frac{p^4q^4}{N^3} + \frac{p^3q^3}{N^2} + \frac{p^2q^2}{N}\right),
\]
where the 1st summand comes from Case 1, the 2nd summand from Case 3, and the 3rd summand from Case 2. Dividing this rate by $(pq)^2$, since we are interested in the variance of $\frac{1}{pq}\mathrm{tr}(\widetilde C^2)$, and noting that $\mathcal{O}(pq)=\mathcal{O}(N)$
finishes the variance calculations.

Putting the mean and variance behavior together using a standard Chebyshev-type argument yields
\[
T_N = \widetilde\Delta_N^2 + \frac{pq}{N} + o_p(1),
\]
in the case of John's or Nagao's statistic, while
\[
T_N = \widetilde\Delta_N^2 + o_p(1),
\]
in the case of the Ledoit-Wolf one. Notice that the probability of ties in the simulated statistics is zero, which is why we are allowed to keep the number of Monte Carlo runs $M$ fixed. This yields consistency and the proof is complete.

\subsection{Proof of Proposition~\ref{prop:angular_level}}

Similarly to the whitened test, the goal is again to show that $T_N^{(A)}(\{X_n\}) \sim T_N^{(A)}(\{z_n\})$ and hence the Monte Carlo simulation correctly calibrates the test up to the approximation error.

Consider the spherical projection operator $\Pi(y) = y/\|y\|$. Since orthogonal transformations preserve Frobenius norms, the spherical projections does not break equivariance since
\[
\Pi(Qy) = \frac{Q y}{\|Qy\|}=Q \frac{y}{\|y\|} = Q\, \Pi(y).
\]

Using the same notation as in the proof of Theorem~\ref{thm:law}, we have
\[
\vec\big((\mmle)^{-1/2} X_n \big) = Q(A^{-1/2}z_n).
\]
Applying the spherical projection operator to both sides of the previous equation, we obtain
\[
u_n := \vec(U_n) = \Pi(Q A^{-1/2}z_n) = Q \Pi(A^{-1/2}z_n) =: Q v_n,
\]
where $U_n$ is from Definition~\ref{def:angular_test}. Thus the empirical covariance based on the matrix-whitened and sphere-projected sample $\{u_n\}$ obtained from $\{X_n\}$ has the same eigenvalues as the empirical covariance based on the matrix-whitened and sphere-projected sample $\{v_n\}$ obtained from $\{z_n\}$. Since the angular test only depends on the eigenvalues of such empirical covariances, we have $T_N^{(A)}(\{X_n\})\sim T_N^{(A)}(\{z_n^\star\})$, where $\{z_n^\star\}$ generated from the distribution of $\{z_n\}$.

\subsection{Proof of Lemma~\ref{lem:angular_computation}}

The proof is straightforward. By definition:
\[
T_N := T_N^{(A)} = pq \left\| \widehat{S} - \frac{1}{pq} I_{p,q}\right\|_F^2,
\]
where
\[
\widehat{S} = \frac{1}{N}\sum_{n=1}^N U_n \otimes U_n \quad \text{and} \quad U_n = \frac{(\widehat{\Sigma}_1 \ot \widehat{\Sigma}_2)^{-1/2}X_n}{\| (\widehat{\Sigma}_1 \ot \widehat{\Sigma}_2)^{-1/2}X_n \|_F}.
\]
Because \(\|U_n\|_F=1\), we have
\[
\operatorname{tr}(\widehat S)
=
\frac{1}{N}\sum_{n=1}^N \operatorname{tr}(U_n\otimes U_n)
=
\frac{1}{N}\sum_{n=1}^N \|U_n\|_F^2
=
1.
\]
Therefore,
\[
\begin{aligned}
T_N
&=
pq\left\|
\widehat S-\frac{1}{pq}I_{p,q}
\right\|_F^2 =
pq\left\{
\operatorname{tr}(\widehat S^2)
-\frac{2}{pq}\operatorname{tr}(\widehat S)
+\frac{1}{p^2q^2}\operatorname{tr}(I_{p,q})
\right\} \\
&=
pq\left\{
\operatorname{tr}(\widehat S^2)
-\frac{2}{pq}
+\frac{1}{pq}
\right\} =
pq\,\operatorname{tr}(\widehat S^2)-1.
\end{aligned}
\]
Moreover,
\[
\begin{aligned}
\operatorname{tr}(\widehat S^2)
&=
\langle \widehat S,\widehat S\rangle_F \\
&=
\frac{1}{N^2}
\sum_{n=1}^N
\sum_{n'=1}^N
\left\langle
U_n\otimes U_n,\,
U_{n'}\otimes U_{n'}
\right\rangle_F \\
&=
\frac{1}{N^2}
\sum_{n=1}^N
\sum_{n'=1}^N
\langle U_n,U_{n'}\rangle_F^2,
\end{aligned}
\]
which completes the proof.

\subsection{Proof of Proposition~\ref{prop:relevant_hypothesis}}

Because \(Z\) is spherical, we may write \(Z=RV\), where \(V\) is uniform on the unit sphere \(\mathbb S^{d-1}\) and \(R\ge 0\) is independent of \(V\). The radial part cancels under normalization, so
\[
	\frac{Y}{\|Y\|_F}
	\sim
	\frac{C^{1/2}V}{\langle V, C V\rangle^{1/2}}.
\]
Consider the matrix-vector version of the previous formula, i.e.
\[
	\frac{y}{\|y\|}
	\sim
	\frac{\mat(C)^{1/2}v}{\sqrt{v^\top \mat(C) v}},
\]
and the eigendecomposition $\mat(C) = Q \Lambda Q^\top$. We have
\[
\mat(S)
	=
	\E\left[
	\frac{\mat(C)^{1/2}vv^\top \mat(C)^{1/2}}{v^\top \mat(C) v}
	\right]
    = \E\left[ \frac{Q \Lambda^{1/2}Q^\top v v^\top Q \Lambda^{1/2}Q^\top}{v^\top Q \Lambda Q^\top v} \right].
\]
Since $V$ is uniform and thus rotation-invariant, we can replace $Q^\top v$ by $v$ everywhere to obtain
\[
\mat(S) = Q \E\left[ \frac{\Lambda^{1/2}v v^\top \Lambda^{1/2}}{v^\top \Lambda v} \right] Q^\top.
\]
It is thus enough to deal with the diagonal case $\mat(C)=\Lambda$ for which
\[
\mat(S) = \E\left[ \frac{\Lambda^{1/2}v v^\top \Lambda^{1/2}}{\sum_{k=1}^{pq} \lambda_k v_k^2} \right].
\]
By symmetry under sign changes of the coordinates of $V$, the off-diagonal entries vanish. Thus $S$ must be diagonal:
	\[
	\mat(S)=\mathrm{diag}(s_1,\dots,s_{pq})
	\quad\text{with}\quad
	s_j
	=
	\E\!\left[
	\frac{\lambda_j v_j^2}{\sum_{k=1}^d \lambda_k v_k^2}
	\right].
	\]
	
We claim that for \(i\neq j\),
	\[
	(\lambda_i-\lambda_j)(s_i-s_j)\ge 0,
	\]
	with strict inequality whenever \(\lambda_i\neq \lambda_j\).
Indeed, write
	\[
	a:=v_i^2,\qquad b:=v_j^2,\qquad D:=\sum_{k\neq i,j}\lambda_k v_k^2.
	\]
	Then
	\[
	s_i-s_j
	=
	\E\!\left[
	\frac{\lambda_i a-\lambda_j b}{\lambda_i a+\lambda_j b+D}
	\right].
	\]
	Using exchangeability of the coordinates of \(v\), we may average this integrand
	with the one obtained by swapping \(a\) and \(b\), which gives
	\begin{equation}\label{eq:eigdiff}
	s_i-s_j
	=
	\frac12\E\!\left[
	\frac{\lambda_i a-\lambda_j b}{\lambda_i a+\lambda_j b+D}
	+
	\frac{\lambda_i b-\lambda_j a}{\lambda_i b+\lambda_j a+D}
	\right].
	\end{equation}
	A direct algebraic simplification yields
	\[
	\frac{\lambda_i a-\lambda_j b}{\lambda_i a+\lambda_j b+D}
	+
	\frac{\lambda_i b-\lambda_j a}{\lambda_i b+\lambda_j a+D}
	=
	(\lambda_i-\lambda_j)
	\frac{D(a+b)+2ab(\lambda_i+\lambda_j)}
	{(\lambda_i a+\lambda_j b+D)(\lambda_i b+\lambda_j a+D)}.
	\]
	The fraction on the right is nonnegative, and strictly positive almost surely.
	Hence
	\begin{equation}\label{eq:proof4}
	\mathrm{sign}(s_i-s_j)=\mathrm{sign}(\lambda_i-\lambda_j).
	\end{equation}
	
	Now suppose
	\[
	\mat(S)=\frac1d I_{pq}.
	\]
	Then all diagonal entries of $\mat(S)$ are equal, i.e.
	\[
	s_1=\cdots=s_{pq}.
	\]
	By~\eqref{eq:proof4}, this implies
	\[
	\lambda_1=\cdots=\lambda_{pq}.
	\]
	Since \(\tr(C)=\tr(\Lambda)=d\), the common value must be \(1\). Thus $
	\Lambda=I_{pq}$, and $C=I_{p,q}$.
	
	Conversely, if $C=I_{p,q}$, then $Y=Z$ is spherical, so $Y/\|Y\|_F$ is uniform on
	the unit sphere and therefore
	\[
	S=\E\left[\frac{Z \otimes Z}{\|Z\|_F^2}\right]=\frac1{pq} I_{p,q}.
	\]
	
	This proves that $S=I_{p,q}/pq$ if and only if $C=I_{p,q}$, which is if and only if $\Sigma=\Sigma_1 \ot \Sigma_2$ for some $\Sigma_1$ and $\Sigma_2$ by the proof of Theorem~\ref{thm:consistency}. This proves the equivalence.

    It now remains to prove the relevant hypothesis statement. To this end, note that by assumption (B2) the eigenvalues of $C$ are uniformly bounded from above and below, and
    \[
    \frac{1}{pq}\|C-I_{p,q}\|^2_F \geq \widetilde{\Delta}^2_0 > 0
    \]
    exactly the same as in the proof of Theorem~\ref{thm:consistency}.

    It remains to show that the separation is uniform in the growing dimension
\(d=pq\). Consider again formula \eqref{eq:eigdiff}.
Since \(D\ge m(1-a-b)\) and both denominators are bounded above by \(M\), we have
\[
|s_i-s_j|
\ge
\frac{|\lambda_i-\lambda_j|}{2}
\frac{m}{M^2}
\mathbb E[(a+b)(1-a-b)].
\]
Now \(a+b\sim \mathrm{Beta}(1,(d-2)/2)\) by \cite[Theorems 1.4 and 1.5]{FKN}, and so \cite[Theorems 1.3]{FKN}
\[
\mathbb E[(a+b)(1-a-b)]
=
\frac{2(d-2)}{d(d+2)}
\ge
\frac{2}{5d}
\]
for all \(d\ge 3\). Hence
\[
|s_i-s_j|
\ge
\frac{c_0}{d}|\lambda_i-\lambda_j|,
\qquad
c_0:=\frac{m}{5M^2}.
\]
Therefore,
\[
2d\left\|S(C)-\frac1dI_d\right\|_F^2
=
\sum_{i,j=1}^d(s_i-s_j)^2
\ge
\frac{c_0^2}{d^2}
\sum_{i,j=1}^d(\lambda_i-\lambda_j)^2
=
\frac{2c_0^2}{d}\|C-I_d\|_F^2.
\]
Thus
\[
d\left\|S(C)-\frac1dI_d\right\|_F^2
\ge
c_0^2\frac1d\|C-I_d\|_F^2.
\]
By the whitened separation obtained from Assumption (B5) in the proof of Theorem~\ref{thm:consistency}, we have
\(d^{-1}\|C-I_d\|_F^2\ge \widetilde\Delta_0^2\), and hence
\[
d\left\|S(C)-\frac1dI_d\right\|_F^2
\ge
c_0^2\widetilde\Delta_0^2
=:\wideparen\Delta_0^2>0.
\]

\subsection{Proof of Theorem~\ref{thm:angular_consistency}}

Let us briefly recall the notation. For $n=1,\ldots,N$, the data matrices $X_n \in \R^{p\times q}$ are zero mean with covariance $\Sigma \in \R^{p \times q \times p \times q}$, i.e.~$X_n = \Sigma^{1/2} Z_n$ with $Z_n$ having mean zero, variance one and uncorrelated entries. The matrix whitened data matrices are denoted $Y_n = (\mmle)^{-1/2} X_n$, and we also need the oracle-whitened version $\widetilde{Y}_n = (\Sigma_1 \ot \Sigma_2)^{-1/2} X_n$, which has the covariance $C$. The sphere-projected matrices $U_n= Y_n / \|Y_n\|_F$, $n=1,\ldots,N$, yield the empirical covariance $\widehat{S}$ from Definition~\ref{def:angular_test}. Finally, we will also need their oracle versions $\widetilde{U}_n=\widetilde{Y}_n/\|\widetilde{Y}_n\|_F$, whose population covariance is denoted by $S \in \R^{p \times q \times p \times q}$. The data matrices are denoted with capital letters, while their vectorizations will be denoted with lowercase letters.

Before we prove the theorem itself, we provide one more lemma.

\begin{lemma}\label{lem:rates}
    Under Assumptions (B1)-(B3) of Theorem~\ref{thm:angular_consistency}, we have
\[
	\|S\|_{2}=O\big((pq)^{-1}\big),
	\qquad
	\|S\|_F=O\big((pq)^{-1/2}\big),
	\qquad
	\|\widetilde S\|_F=O_p\big(N^{-1/2}\big).
\]
\end{lemma}

\begin{proof}
Firstly, note that by assumption~(B2), there exist constants \(0<m<M<\infty\) such that
\[
mI_{p,q}\preceq \Sigma_{1}\ot\Sigma_{2}\preceq M I_{p,q}
\]
for all $N$ large enough.

Secondly, we will use a re-normalization trick to re-express $\widetilde{U}_n$ as
\[
\widetilde{U}_n = \frac{(\Sigma_1 \ot \Sigma_2)^{-1/2} \Sigma^{1/2} Z_n}{\| (\Sigma_1 \ot \Sigma_2)^{-1/2} \Sigma^{1/2} Z_n\|_F} \frac{\|Z_n\|_F}{\|Z_n\|_F} =: \frac{C^{1/2} V_n}{\sqrt{\langle V_n, C V_n\rangle}}, 
\]
where $V_n$ is uniform on the Frobenius unit sphere in $\R^{p\times q}$. Therefore,
\[
S = \E\left[ \frac{C^{1/2}V \otimes C^{1/2}V}{\langle V, C V\rangle}\right],
\]
and we have for any unit-norm $X \in \R^{p \times q}$
\[
\begin{split}
\langle X, S X\rangle = \langle X \otimes X, S\rangle &\leq \frac{1}{m} \langle X \otimes X , (C \ot C)^{1/2} \E[V \otimes V] \rangle  \\
&\leq \frac{1}{m} \langle (C \ot C)^{1/2} (X \otimes X) ,  \frac{1}{pq} I_{p}\ot I_q \rangle\\
&= \frac{1}{mpq} \tr(C^{1/2}X \otimes C^{1/2}X) = \frac{1}{mpq} \langle C^{1/2}X, C^{1/2}X \rangle \\
&=\frac{1}{mpq} \langle X,CX \rangle \leq \frac{M}{mpq},
\end{split}
\]
where in the first equality we used that the trace of a rank one operator is the inner product as in $\tr(U \otimes V) = \langle U,V\rangle$.
Therefore, $\|S\|_{2} = \mathcal O\big( (pq)^{-1} \big)$, which is the first claim.

As for the second claim, since $\tr(S) = \E \|\widetilde{U}\|_F^2=1$, we obtain immediately that
\[
\|S\|_F^2=\tr(S^2)\le \|S\|_{2}\tr(S)=\|S\|_{2}=\mathcal O\big( (pq)^{-1} \big),
\]
and so it is $\|S\|_F = \mathcal O\big( (pq)^{-1/2} \big)$.

For the final claim, note that
\[
\| \widetilde S\|_F^2 = \tr(\widetilde S^2) = \frac{1}{N^2} \sum_{n,m=1}^N \langle \widetilde{U}_n,\widetilde{U}_m \rangle^2.
\]
Taking expectation and utilizing independence for $n \neq m$, we obtain
\[
\E \tr(\widetilde{S}^2) = \frac{N(N-1)}{N^2}\tr(S^2) + \frac{1}{N}.
\]
But by the previous bound and our high-dimensional scaling $\tr(S^2) = \mathcal{O}(N^{-1})$, and so
\[
\E \|\widetilde{S}\|_F^2 = \mathcal{O}(N^{-1}).
\]
Finally, boundedness of expectation implies boundedness in probability for non-negative random variables such as $\|\widetilde{S}\|_F^2$ by Markov's inequality, yielding the desired rate:
\[
\|\widetilde{S}\|_F = \mathcal{O}(N^{-1/2}).
\]
\end{proof}

We are now ready to show Theorem~\ref{thm:angular_consistency} following a similar path as in the proof of Theorem~\ref{thm:consistency}. Firstly, Proposition~\ref{prop:relevant_hypothesis} yields that dense departure from separability guarantees dense departure from sphericity on the matrix-whitened level, i.e.
\begin{equation}\label{eq:departure}
\wideparen{\Delta}_N^2 := pq \| S - \frac{1}{pq} I_{pq} \|^2_F \geq \wideparen{\Delta}_0^2 > 0
\end{equation}
for $N$ large enough.

Consider the centered statistic, introduce the oracle and treat the resulting two terms separately:
\[
	T_N^{(A)}-\wideparen{\Delta}_N^2
	=
	pq\,\tr(\widehat S^2-S^2)
	=
	pq\,\tr(\widehat S^2-\widetilde S^2)
	+
	pq\,\tr(\widetilde S^2-S^2)
	=: (I)+(II).
\]
If we show that $(I)$ and $(II)$ are both $o_p(1)$ up to an additive constant, we will have the Monte Carlo reference distribution concentrated around that additive constant, while the observed statistic under the alternative will be separated from this null center by at least $\wideparen{\Delta}_0^2$, yielding consistency of the test.

Let us begin by showing that $(I)=o_p(1)$. By Assumption~(B2), there exist constants \(0<m<M<\infty\) such that
\begin{equation}\label{eq:spectral_control}
mI_{p,q}\preceq \Sigma_{1}\ot\Sigma_{2}\preceq M I_{p,q}
\end{equation}
for all $N$ large enough. Since, by Assumption~(B4),
\[
\|\widehat\Sigma_{1}\ot\widehat\Sigma_{2}-\Sigma_{1}\ot\Sigma_{2}\|_{2}=o_p(1),
\]
Weyl's inequality implies that, with probability tending to one,
\[
\frac m2 I_{d_N}\preceq \widehat\Sigma_{1}\ot\widehat\Sigma_{2}\preceq 2M I_{d_N}.
\]
Hence both $\widehat\Sigma_{1}\ot\widehat\Sigma_{2}$ and $\Sigma_{1}\ot\Sigma_{2}$ belong, eventually with high probability, to the same compact subset of the positive-definite cone. On this set, the square-root inverse map
\[
f:\mathcal S_{++}^{d_N}\to \mathcal S_{++}^{d_N},
\qquad
f(X)=X^{-1/2},
\]
is Lipschitz in operator norm. Therefore
\[
\|\widehat\Sigma_{1}^{-1/2}\ot\widehat\Sigma_{2}^{-1/2}-\Sigma_{1}^{-1/2}\ot\Sigma_{2}^{-1/2}\|_{2}=o_p(1),
\]
Therefore, denoting $G_N = \mat\big((\mmle)^{-1/2} (\Sigma_{1}\ot\Sigma_{2})^{1/2}\big)$ and $E_N=G_N - I_{pq}$, we have
\[
\|E_N\|_{2} = 
o_p(1).
\]

Now, write
\[
y_n=G_N\widetilde y_n,
\qquad
u_n=\frac{G_N\widetilde u_n}{\|G_N\widetilde u_n\|},
\]
and define
\[
\widetilde P_n:=\widetilde u_n\widetilde u_n^\top,
\qquad
P_n:=u_n u_n^\top.
\]
We can write
\[
P_n = u_nu_n^\top
= \frac{G_N\widetilde u_n\widetilde u_n^\top G_N^\top}{\|G_N\widetilde u_n\|^2}
= \frac{G_N\widetilde P_nG_N^\top}{\|G_N\widetilde u_n\|^2},
\]
where the denominator can be expressed as
\[
\|G_N\widetilde u_n\|^2
= \widetilde u_n^\top(I+E_N)^\top(I+E_N)\widetilde u_n
= 1+\widetilde u_n^\top\!\bigl(E_N+E_N^\top+E_N^\top E_N\bigr)\widetilde u_n
= 1+\alpha_n
\]
and the numerator expands as
\[
G_N\widetilde P_nG_N^\top
= \widetilde P_n + E_N\widetilde P_n + \widetilde P_nE_N^\top + E_N\widetilde P_nE_N^\top.
\]
Subtracting $\widetilde P_n = (1+\alpha_n)\widetilde P_n/(1+\alpha_n)$ yields
\[
P_n-\widetilde P_n
=
\frac{
	E_N\widetilde P_n+\widetilde P_nE_N^\top+E_N\widetilde P_nE_N^\top
	-\alpha_{n}\widetilde P_n
}{1+\alpha_{n}},
\]
where
\[
\alpha_{n}
=
\widetilde u_n^\top(E_N+E_N^\top+E_N^\top E_N)\widetilde u_n.
\]

Since \(\|E_N\|_{2}=o_p(1)\) and $\widetilde u_n$ have norms one, we have
\[
\sup_{1\le n\le N}|\alpha_{n}|=o_p(1),
\]
and thus
\[
\sup_{1\le n\le N}\frac1{|1+\alpha_{n}|}=\mathcal O_p(1).
\]
Averaging the projector identity over $N$, we obtain
\[
\|\widehat S-\widetilde S\|_F
\le
\mathcal O_p(1) \Bigg(
\|E_N\widetilde S\|_F
+
\|\widetilde S E_N^\top\|_F
+
\|E_N\widetilde S E_N^\top\|_F
+
\left\|
\frac1N\sum_{n=1}^N \alpha_{n}\widetilde P_n
\right\|_F
\Bigg).
\]

By Lemma~\ref{lem:rates}, $\|\widetilde S\|_F=\mathcal O_p(d_N^{-1/2})$, and thus
\[
\|E_N\widetilde S \|_F
\le
\|E_N\|_{2}\|\widetilde S \|_F
=
o_p(1)\,\mathcal O_p(d_N^{-1/2})
=
o_p(d_N^{-1/2}),
\]
and similarly
\[
\|\widetilde S E_N^\top\|_F=o_p\big((pq)^{-1/2}\big)
\qquad \& \qquad
\|E_N\widetilde S E_N^\top\|_F
\le
\|E_N\|_{2}^2\|\widetilde S \|_F
=
o_p\big((pq)^{-1/2}\big).
\]
Furthermore, denoting for short
\[
B_N:=E_N+E_N^\top+E_N^\top E_N,
\]
it is $\alpha_{n}=\widetilde u_n^\top B_N\widetilde u_n$ and thus
\[
\left\|
\frac1N\sum_{n=1}^N \alpha_{n}\widetilde P_n
\right\|_F
\le
\|B_N\|_{2}\,\|\widetilde S\|_F
=
o_p(1)\,\mathcal O_p\big((pq)^{-1/2}\big)
=
o_p\big((pq)^{-1/2}\big).
\]
Overall, we obtain $\|\widehat S-\widetilde S\|_F=o_p\big((pq)^{-1/2}\big)$.

Finally,
\[
|(I)|
=
pq\bigl|\tr((\widehat S -\widetilde S )(\widehat S +\widetilde S ))\bigr|
\le
pq\|\widehat S -\widetilde S \|_F
\bigl(\|\widehat S \|_F+\|\widetilde S \|_F\bigr).
\]
Since
\[
\|\widehat S \|_F
\le
\|\widehat S -\widetilde S \|_F+\|\widetilde S \|_F
=
O_p\big((pq)^{-1/2}\big),
\]
we conclude that $(I)=o_p(1)$.

Now let us treat the remaining term and show that $(II)=pq/N + o_p(1)$. Denoting $\Xi := \widetilde S - S$, we have
\[
\tr(\widetilde S^2)-\tr(S^2)
=
2\tr(S \Xi)+\tr(\Xi^2),
\]
and so
\begin{equation}\label{eq:proof1}
(II)
=
2pq\tr(S \Xi)+pq\|\Xi\|_F^2.
\end{equation}

Let us treat the linear term in equation~\eqref{eq:proof1} first. We can write
\[
2pq\tr(S \Xi) = \frac{2pq}{N} \sum_{n=1}^N \left( \langle \widetilde{U}_n, S \widetilde{U}_n\rangle - \tr(S^2)\right).
\]
Since $\E \langle \widetilde{U}_n, S \widetilde{U}_n\rangle = \langle \E[\widetilde{U}_n \otimes \widetilde{U}_n], S \rangle=\langle S, S \rangle$, the summands are centered and i.i.d. Moreover, by Lemma~\ref{lem:rates}, we have
\[
0\leq \langle \widetilde{U}_n, S \widetilde{U}_n\rangle \leq \|S\|_{2}=\mathcal O\big((pq)^{-1}\big)
\quad \& \quad
\tr(S^2)=\mathcal O\big((pq)^{-1}\big),
\]
so the variance of the summands is $\mathcal{O}\big( (pq)^{-2} \big)$ and thus
\[
\mathrm{var}\left(
\frac{2pq}{N} \sum_{n=1}^N \left( \langle \widetilde{U}_n, S \widetilde{U}_n\rangle - \tr(S^2)\right)
\right)
=
\mathcal O(N^{-1}),
\]
which implies by Chebyshev's inequality that the linear terms converges to zero:
\[
	2pq\,\tr(S\Xi)=o_p(1).
\]

As for the quadratic term, denote
\[
\xi_n := \widetilde U_n \otimes \widetilde U_n - S, \quad \text{so that} \quad \Xi = \frac{1}{N}\sum_{n=1}^N \xi_n
\]
and
\begin{equation}\label{eq:proof2}
pq\|\Xi\|_F^2
=
\frac{pq}{N^2}\sum_{n=1}^N \|\xi_n\|_F^2
+
\frac{2pq}{N^2}\sum_{1\le n<m\le N}\langle \xi_n, \xi_m \rangle.
\end{equation}

For the diagonal term in~\eqref{eq:proof2}, note that
\[
\| \xi_n \|_F^2 = \| \widetilde U_n \otimes \widetilde U_n \|_F^2 - 2 \langle \widetilde U_n, S \widetilde U_n \rangle + \| S \|_F^2.
\]
Since $\|\widetilde U_n\|_F=1$, the first summand is equal to one. Secondly $\E \langle \widetilde U_n, S \widetilde U_n \rangle = \|S\|_F^2$ as before, and so
\[
\E\|\xi_n\|_F^2=1-\|S\|_F^2.
\]
At the same time, $\| \xi_n \|_F \leq 1 + \|S\|_F$ and so by the law of large numbers we have
\[
\frac{pq}{N^2}\sum_{n=1}^N \|\xi_n\|_F^2
	=
	\frac{pq}{N}\bigl(1-\|S\|_F^2\bigr)+o_p(1)
	=
	\frac{pq}{N}+o_p(1),
\]
because $\|S\|_F^2 = o_p(1)$ by Lemma~\ref{lem:rates}.

For the off-diagonal term in~\eqref{eq:proof2}, define the kernel
\[
k(U,V) = \langle U \otimes U - S, V \otimes V - S \rangle.
\]
Then
\[
\frac{2pq}{N^2}\sum_{1\le n<m\le N}\tr(\xi_n \xi_m)
=
\frac{2pq}{N^2}\sum_{1\le n<m\le N} k(\widetilde U_n,\widetilde U_m).
\]
The kernel is degenerate since
\[
\E[k(U,V)|V] = \langle \E[U \otimes U] - S, V \otimes V - S \rangle = 0.
\]
Expand the kernel as in
\[
k(U,V) = \langle U,V\rangle^2 - \langle U, S U \rangle - \langle V, S V \rangle + \|S\|_F^2,
\]
and bound its square by
\begin{equation}\label{eq:proof3}
k(U,V)^2 \leq 4 \big[ \langle U,V\rangle^4 + \langle U, S U \rangle^2 + \langle V, S V \rangle^2 + \|S\|_F^4 \big].
\end{equation}
If we show that $\E[k(U,V)^2] = \mathcal{O}(N^{-2})$, we will have the off-diagonal term in~\eqref{eq:proof2} being $o_p(1)$, leading to $(II)=pq/N+o_p(1)$, which will complete the proof.

Consider thus the four summands in~\eqref{eq:proof3}. The last one is $\mathcal{O}(N^{-2})$ by Lemma~\ref{lem:rates}. For the second one,
\[
\langle U, S U \rangle^2 \leq \|S\|_{2}^2 \|U\|_F^4 = \|S\|_{2}^2
\]
which is again $\mathcal{O}(N^{-2})$ by Lemma~\ref{lem:rates}, and similarly for the third summand. It thus remains to bound $\E[\langle U,V\rangle^4]$, where $U$ and $V$ are i.i.d., distributed as $\widetilde{U}_1$.

To do that, we utilize again the renormalization trick from the proof of Lemma~\ref{lem:rates}. Specifically, we can write
\[
U = \frac{C^{1/2} A}{ \langle A, C A \rangle^{1/2}} \quad \& \quad V = \frac{C^{1/2} B}{ \langle B, C B \rangle^{1/2}}
\]
where $A$ and $B$ are independent random variables distributed uniformly on the sphere, and $C$ is the matrix-whitened version of the covariance (i.e.~whitening can be naturally done on the sphere as well).

Now, due to~\eqref{eq:spectral_control}, we have
\[
\langle U , V \rangle^4 \leq m^{-4} \langle CA,B\rangle^4.
\]
Denoting $a=\vec(CA)$ and $b=\vec(B)$, we seek to bound $(a^\top b)^4$. Take an orthogonal matrix $Q$ such that $Q^\top a = \|a\| e_1$, where $e_1$ denotes the first standard basis vector. Then
\[
a^\top b = \|a\|e_1^\top Q^\top b =: \|a\|e_1^\top \widetilde{b} = \|a\|\widetilde{b}_1,
\]
where $\widetilde{b}$ is still uniform on the sphere. Hence
\[
\E[(a^\top b)^4|a] = \|a\|^4\E \widetilde b_1^4 = \frac{3 \|a\|^4}{pq(pq+2)},
\]
where the last equality is due to Theorem 1.3 of \citet{FKN}. Moreover
\[
\|a\| = \| CA \| \leq \|C\|_{2} \|A\|_F \leq M,
\]
so $\langle U , V \rangle^4 = \mathcal O\big((pq)^{-2}\big)$, which is $\mathcal O(N^{-2})$ in the high-dimensional scaling of Assumption~(B1). This completes the proof of Theorem~\ref{thm:angular_consistency}.

\section{Oracle Angular vs.~Elliptical Statistics}
\label{sec:oracle-angular-robustness}

The theoretical results about the angular test place it on the same footing as
the elliptical test from Section~3, at least from the perspective of dense high-dimensional
alternatives. The main reason for introducing the angular statistic, however, is different:
it is meant to reduce the sensitivity of the test to the radial component of an elliptical
distribution. In this subsection we isolate this mechanism at the oracle-whitened level.

Below, we follow the notation from the proofs above.

Under the hypothesis of separability \(H_0\), consider the oracle-whitened observations
\[
        \widetilde Y_n
        =
        (\Sigma_1\widetilde\otimes\Sigma_2)^{-1/2}X_n .
\]
Under the elliptical model, \(\widetilde Y_n\) has covariance proportional to the identity
and can be written as a radial variable times a spherical direction. The angular oracle
observations are then
\[
        \widetilde U_n
        =
        \frac{\widetilde Y_n}{\|\widetilde Y_n\|_F},
        \qquad n=1,\ldots,N,
\]
whenever \(\widetilde Y_n\neq 0\). Define the oracle angular statistic
\[
        \widetilde T_N^{(A)}
        =
        \frac{pq}{N^2}
        \sum_{n=1}^N\sum_{m=1}^N
        \bigl\langle \widetilde U_n,\widetilde U_m\bigr\rangle_F^2
        -1 .
\]
The following proposition showcases radial invariance of the oracle angular statistic.

Suppose that
\[
        Z_n = R_n V_n ,
        \qquad n=1,\ldots,N,
\]
where \(R_n\geq 0\), the directions \(V_n\) are i.i.d. uniform on the Frobenius unit sphere,
and \(R_n\) is independent of \(V_n\). Under the oracle whitening, we also have
\(\widetilde Y_n=R_n V_n\), that is, the whitened observations have the same directions.
Radial normalization, which is part of the calculation of the angular statistic, then gives
\[
        \widetilde U_n
        =
        \frac{\widetilde Y_n}{\|\widetilde Y_n\|_F}
        =
        \frac{R_n V_n}{R_n\|V_n\|_F}
        =
        V_n,
\]
whenever \(R_n>0\). Thus the radial variables are removed exactly and the distribution of
\(\widetilde T_N^{(A)}\) depends only on the uniform spherical law of the variables
\(\{V_n\}_{n=1}^N\).

This can be contrasted with the elliptical statistic. At the oracle-whitened level, the latter is
\[
        \widetilde T_N
        =
        \frac{1}{N^2pq}
        \sum_{n=1}^N\sum_{m=1}^N
        \langle \widetilde Y_n,\widetilde Y_m\rangle_F^2
        -1 .
\]
Unlike \(\widetilde T_N^{(A)}\), this statistic retains radial information and its null
distribution therefore depends on the radial law.

\begin{proposition}
\label{prop:oracle-elliptical-radial-sensitivity}
Assume \(H_0\) and suppose that the oracle-whitened observations satisfy
\[
        \widetilde Y_n = R_n V_n,
        \qquad n=1,\ldots,N,
\]
where \(V_n\) are i.i.d.~uniform on \(\mathbb S^{pq-1}\), \(R_n\) is independent of \(V_n\), and
the normalization is chosen so that
\[
        \mathbb E\; \widetilde Y_n\otimes \widetilde Y_n = I_{p,q}.
\]
Equivalently, \(\mathbb E R_n^2=pq\). If \(\mathbb E R_n^4<\infty\), then
\[
        \mathbb E \widetilde T_N
        =
        \frac{N-1}{N}
        +
        \frac{1}{N}\,
        \frac{\mathbb E R_n^4}{pq}
        -1 .
\]
In particular, if
\[
        \kappa_{F_0}
        :=
        \frac{\mathbb E R_n^4}{pq(pq+2)},
\]
then
\[
        \mathbb E \widetilde T_N
        =
        -\frac{1}{N}
        +
        \frac{pq+2}{N}\,\kappa_{F_0}.
\]
\end{proposition}

\begin{proof}
For \(n\neq m\), independence and the covariance normalization give
\[
        \mathbb E\langle \widetilde Y_n,\widetilde Y_m\rangle_F^2
        =
        \operatorname{tr}
        \left\{
        \mathbb E(\widetilde Y_n\otimes \widetilde Y_n)
        \,
        \mathbb E(\widetilde Y_m\otimes \widetilde Y_m)
        \right\}
        =
        \operatorname{tr}(I_{p,q})
        =
        pq .
\]
For \(n=m\),
\[
        \langle \widetilde Y_n,\widetilde Y_n\rangle_F^2
        =
        \|\widetilde Y_n\|_F^4
        =
        R_n^4 .
\]
Therefore
\[
\begin{aligned}
        \mathbb E \widetilde T_N
        &=
        \frac{1}{N^2pq}
        \left\{
        N\,\mathbb E R_n^4
        +
        N(N-1)pq
        \right\}
        -1  \\
        &=
        \frac{N-1}{N}
        +
        \frac{1}{N}\frac{\mathbb E R_n^4}{pq}
        -1 .
\end{aligned}
\]
The displayed expression in terms of \(\kappa_{F_0}\) follows from
\[
        \mathbb E R_n^4=pq(pq+2)\kappa_{F_0}.
\]
\end{proof}

For the Gaussian distribution, \(R_n^2\sim\chi^2_{pq}\), so
\[
        \mathbb E R_n^4=pq(pq+2),
\]
and hence \(\kappa_{F_0}=1\). For heavier-tailed elliptical distributions,
\(\kappa_{F_0}\) is larger, and the oracle elliptical null distribution is shifted relative
to the Gaussian reference distribution.

This showcases the oracle mechanism behind the angular test. Radial normalization removes
the elliptical radial component before the statistic is formed, while the non-angular statistic
retains that component even after exact oracle whitening. Thus Gaussian calibration is
intrinsically fragile for the non-angular statistic under heavy-tailed elliptical laws.

\section{Proportionality of Covariance Maps}
\label{sec:covariance_map}

The MMLE defined in Section~\ref{sec:background} provides a single particular procedure of producing a separable covariance. On the population level, the MMLE solves equations given in \eqref{eq:pseudo-parameter}, cf.~\citet{dutilleul1999mle}. On the other hand, the best separable approximation is another procedure of producing a separable covariance, defined on the population level as
\[
(\widehat \Sigma_1, \widehat \Sigma_2) := \argmin_{\Sigma_1,\Sigma_2} \| \Sigma - \Sigma_1 \ot \Sigma_2 \|_F^2,
\]
cf.~\citet{van1993approximation}. Let us call any such procedure on the population level a \emph{separable covariance map}. For a separable covariance map $\Omega$ and a specific distribution $F_X$ of the random variable $X$, $\Omega(F_X) = \Omega_1(F_X) \ot \Omega_2(F_X)$ is a separable population target.

The following lemma states that, under separability, the specific choice of a separable covariance map is immaterial, as long as the map produces a separable and matrix affine equivariant population target. This is a straightforward generalization of a similar result in multivariate statistics \citep[e.g.][]{bilodeau1999theory}.

\begin{lemma}\label{lem:proportionality of the covariances under ellipticity}
    Let $X$ follow matrix elliptical distribution with a separable covariance $\Sigma_1\ot\Sigma_2$,  and let
    $\Omega(F_X)$ be any matrix affine equivariant separable covariance maps. Then, 
\[
    \Omega(F_X) \,\propto\, \Sigma_1\ot\Sigma_2.
\]
\end{lemma}
\begin{proof}[Proof of Lemma~\ref{lem:proportionality of the covariances under ellipticity}]
    The ellipticity of $X\in\R^{p\times q}$ implies that it admits a stochastic representation $X=\mu+(\Sigma_1\ot\Sigma_2)^{1/2}Z$, where $Z$ is matrix spherical. Take now any $U\in\R^{p\times p}$ and $V\in\R^{q\times q}$ orthogonal. Sphericity of $Z$ implies that $UZV^\top\sim Z$ giving further that 
    $$
    \Omega_1(F_Z)=\Omega_1(F_{UZV^\top})=U\Omega_1(F_Z)U^\top,\qquad \Omega_2(F_Z)=\Omega_2(F_{UZV^\top})=V\Omega_2(F_Z)V^\top,
    $$
    where the second pair of equalities follows from the affine equivariance of $\Omega_1\ot\,\Omega_2$. Thus, $(\Omega_1\ot\Omega_2)(F_Z)\,\propto\, I_{p,q}$. Affine equivariance further gives
    $$
    (\Omega_1\ot \Omega_2)(F_X) = (\Sigma_1\ot\Sigma_2)^{1/2}(\Omega_1\ot \Omega_2)(F_Z)(\Sigma_1\ot\Sigma_2)^{1/2}\,\propto\, \Sigma_1\ot\Sigma_2.
    $$
\end{proof}

As discussed in the main paper, we use the matrix normal MLE~\eqref{eq:mmle} as our whitening matrix choice throughout. For a random sample of sufficiently large size from an absolutely continuous distribution, it is (up to the scale ambiguity) uniquely defined, positive definite, and matrix affine equivariant. Of course, it is not the only such map. For example, using the same techniques as in the proof of Proposition~\ref{prop:MMLE}, one can show that the MLE under the matrix $t$-distribution (see, e.g., \cite{gupta1994new}) is another affine equivariant separable matrix-valued map. However, by Lemma~\ref{lem:proportionality of the covariances under ellipticity}, any two such choices are, on the population level and under $H_0$, mutually proportional and proportional to the separable covariance itself. The choice of whitening matrix is therefore immaterial at the population level, and we adopt the matrix normal MLE for its practical simplicity: it is consistent for the separable covariance up to a scalar factor under the entire elliptical family, sharp analysis of the rates of convergence was developed recently \citep{franks2026near}, and it is easy to compute via the flip-flop algorithm. The usual caveat applies: its finite-sample efficiency is optimal at the normal model and may degrade, relative to other affine equivariant covariance estimators, under heavier-tailed alternatives.

\end{document}